\pdfoutput=1 
\documentclass[a4paper,UKenglish,cleveref, autoref, cleverref, thm-restate]{lipics-v2021}

\renewcommand{\emph}[1]{\textbf{\textit{#1}}}

\usepackage{etoolbox}
\newcommand{\ifconference}[1]{{{\ifx\fullversion\undefined{#1}\fi}\xspace}}
\newcommand{\iffullversion}[1]{{{\ifx\conference\undefined{#1}\fi}\xspace}}

\usepackage{graphicx}  
\usepackage{lipsum}  
\newcommand{\hide}[1]{} 
\usepackage{xspace}
\usepackage{textcomp}
\usepackage{comment} 
\usepackage{verbatim}
\usepackage{url}

\newcommand{\algname}[1]{{\textsf{#1}}} 
\newcommand{\defn}[1]{\emph{\textbf{#1}}} 
\newcommand{\fname}[1]{\textsf{#1}} 
\let \originalleft \left
\let\originalright\right
\renewcommand{\left}{\mathopen{}\mathclose\bgroup\originalleft}
\renewcommand{\right}{\aftergroup\egroup\originalright}

\usepackage{scalerel} 

\usepackage{thmtools}

\newcommand{\whp}[1]{\emph{whp}}

\usepackage{pifont}

\usepackage[shortlabels]{enumitem}

\setlist{topsep=0.3em,itemsep=0.2em,parsep=0.1em,leftmargin=*}

\usepackage{float}
\usepackage[font={small},aboveskip=0.3em, belowskip=0.2em]{caption}

\usepackage[labelfont=bf,list=true,skip=0em]{subcaption}
\usepackage{array}
\newcolumntype{L}[1]{>{\raggedright\let\newline\\\arraybackslash\hspace{0pt}}m{#1}}
\newcolumntype{C}[1]{>{\centering\let\newline\\\arraybackslash\hspace{0pt}}m{#1}}
\newcolumntype{R}[1]{>{\raggedleft\let\newline\\\arraybackslash\hspace{0pt}}m{#1}}
\newcolumntype{B}{>{\bf}c}
\usepackage{rotating}

\usepackage{booktabs} 
\usepackage{multicol,multirow}
\usepackage{longtable} 
\usepackage{supertabular} 
\usepackage{colortbl}
\usepackage{bigstrut}

\newcommand{\myparagraph}[1]{\smallskip\noindent\textbf{#1}~~}

\usepackage[ruled,lined,linesnumbered,noend]{algorithm2e}
\usepackage[noend]{algpseudocode}

\makeatletter
\patchcmd{\@algocf@start}
  {-1.5em}
  {0pt}
  {}{}
\newcommand{\nosemic}{\renewcommand{\@endalgocfline}{\relax}}
\newcommand{\dosemic}{\renewcommand{\@endalgocfline}{\algocf@endline}}

\SetSideCommentLeft
\SetKwInput{notations}{Notations}
\SetKwInput{notes}{Notes}
\SetKwInput{maintains}{Maintains}

\SetKwProg{myfunc}{Function}{}{}
\SetKwFor{parForEach}{ParallelForEach}{do}{endfor}
\SetKwFor{Justrepeat}{Repeat}{}{}
\SetKwFor{mystruct}{Struct}{}{}
\SetKwFor{pardo}{In Parallel:}{}{}

\SetKw{MIN}{min}
\SetKw{MAX}{max}
\SetKw{OR}{or}
\SetKw{AND}{and}

\SetCommentSty{mycommfont}

\usepackage{mdframed}
\mdfdefinestyle{mystyle}{innertopmargin=2pt,innerbottommargin=2pt,skipabove=2pt,skipbelow=2pt}
\mdfdefinestyle{densestyle}{linecolor=framelinecolor,innertopmargin=0,innerbottommargin=0,leftmargin=0,rightmargin=0,backgroundcolor=gray!20}
\mdfdefinestyle{compactcode}{linecolor=framelinecolor,innertopmargin=1pt,innerbottommargin=1pt,backgroundcolor=gray!20,skipabove=0pt,skipbelow=0pt,leftmargin=0,rightmargin=0}

\usepackage{framed}

\usepackage{listings}

\newdimen\zzsize
\newdimen\kwsize
\newcommand{\basicstyle}{\fontsize{\zzsize}{1\zzsize}\ttfamily}
\newcommand{\keywordstyle}{\fontsize{\kwsize}{1\kwsize}\ttfamily\bf}

\newdimen\zzlstwidth
\crefname{section}{Section}{Section}
\crefname{theorem}{Theorem}{Theorem}
\crefname{thm}{Theorem}{Theorem}
\crefname{lemma}{Lemma}{Lemma}
\crefname{corollary}{Corollary}{Corollary}
\crefname{table}{Table}{Table}
\crefname{algorithm}{Algorithm}{Algorithm}
\crefname{figure}{Figure}{Figure}
\crefname{fact}{Fact}{Fact}
\Crefname{table}{Table}{Table}
\crefname{problem}{Problem}{Problem}

\usepackage{tikz} 

\usepackage{titlecaps}

\hide{ 
\binoppenalty=700
\brokenpenalty=0 
\clubpenalty=0   
\displaywidowpenalty=0   
\exhyphenpenalty=50
\floatingpenalty=20000
\hyphenpenalty=50
\interlinepenalty=0
\linepenalty=10
\postdisplaypenalty=0
\predisplaypenalty=0 
\relpenalty=500
\widowpenalty=0  
}

\SetArgSty{normalfont}

\usepackage{cite}
\usepackage{tikz}
\usepackage{tikz-qtree}
\usetikzlibrary{shapes,arrows,positioning}

\def\showcommends{0} 

\ifnum\showcommends=1
  \newcommand{\ryuto}[1]{{\leavevmode\color{red}Ryuto: #1}}
  \newcommand{\yihan}[1]{{\leavevmode\color{blue}Yihan: #1}}
  \newcommand{\yan}[1]{{\leavevmode\color{violet}Yan: #1}}
\else
  \newcommand{\ryuto}[1]{}
  \newcommand{\yihan}[1]{}
  \newcommand{\yan}[1]{}
\fi

\newcommand{\ourmodel}{Fork-Join I/O Model} 

\newcommand{\join}{join}
\newcommand{\funcjoin}{\algname{Join}}
\newcommand{\funcsplit}{\algname{Split}}
\newcommand{\fastunion}{\algname{Union}}
\newcommand{\union}{\algname{Work-Inefficient-Union}}
\newcommand{\bjoin}{\algname{B-Way-Join}}

\newcommand{\intersection}{\algname{Intersection}}
\newcommand{\difference}{\algname{Difference}}

\newcommand{\mjoin}{\algname{Multi-Join}}
\newcommand{\msplit}{\algname{Multi-Split}}
\newcommand{\tsplit}{\algname{Thread-Split}}
\newcommand{\bjoinfast}{\algname{B-Way-Join-Fast}}
\newcommand{\joinwithtall}{\algname{Join-With-Tall}}
\newcommand{\prefixsum}{\algname{Prefix-Sum}}
\newcommand{\dividenode}{\algname{Divide-Node}}

\newcommand{\twobinarysearch}{\algname{Two-Way-Binary-Search}}
\newcommand{\threadgather}{\algname{Gather}}

\newcommand{\treejoinset}{X}

\newcommand{\hmax}{h^*}
\newcommand{\successor}{\fname{succ}}

\title{Parallel Joinable B-Trees in the \ourmodel{}}

\author{Michael T. Goodrich}{University of California, Irvine}{goodrich@uci.edu}{https://orcid.org/0000-0002-8943-191X}{}
\author{Yan Gu}{University of California, Riverside}{ygu@cs.ucr.edu}{https://orcid.org/0000-0002-4392-4022}{}
\author{Ryuto Kitagawa}{University of California, Irvine}{ryutok@uci.edu}{https://orcid.org/0009-0000-7329-9590}{}
\author{Yihan Sun}{University of California, Riverside}{yihans@cs.ucr.edu}{https://orcid.org/0000-0002-3212-0934}{}

\authorrunning{M. Goodrich, Y. Gu, R. Kitagawa, Y. Sun}

\Copyright{Michael Goodrich, Yan Gu, Ryuto Kitagawa, Yihan Sun}


\nolinenumbers 

\EventEditors{Ho-Lin Chen, Wing-Kai Hon, and Meng-Tsung Tsai}
\EventNoEds{3}
\EventLongTitle{36th International Symposium on Algorithms and Computation (ISAAC 2025)}
\EventShortTitle{ISAAC 2025}
\EventAcronym{ISAAC}
\EventYear{2025}
\EventDate{December 7--10, 2025}
\EventLocation{Tainan, Taiwan}
\EventLogo{}
\SeriesVolume{359}
\ArticleNo{39}

\begin{document}

\maketitle
\setcounter{page}{0}
\begin{abstract}
  Balanced search trees are widely used in computer science to efficiently
  maintain dynamic ordered data.  
  To support efficient set operations (e.g., union, intersection, difference) using trees, the \emph{join-based} framework is widely studied. 
  This framework has received particular attention in the parallel setting, 
  and has been shown to be effective in enabling simple and theoretically efficient set operations on trees. 
  Despite the widespread adoption of parallel \join-based trees, a major drawback of previous work on such data structures is the inefficiency of their input/output (I/O) access patterns. 
  Some recent work (e.g., C-trees and PaC-trees) focused on more I/O-friendly implementations of these algorithms. 
  Surprisingly, however, there have been no results on bounding the I/O-costs for these algorithms. It remains open whether these algorithms can provide tight, provable guarantees in I/O-costs on trees. 

  This paper studies efficient parallel algorithms for set operations based on search tree algorithms using a \join-based framework, with a special focus on achieving I/O efficiency in these algorithms. 
  To better capture the I/O-efficiency in these algorithms in parallel, 
  we introduce a new computational model, the \ourmodel{}, to measure the I/O costs in fork-join parallelism. 
  This model measures the total block transfers (I/O work) and their critical path (I/O span). 
  Under this model, we propose our new solution based on B-trees.
  Our parallel algorithm computes the union, intersection, and difference of two
  B-trees with $O(m \log_B(n/m))$ I/O work
  and $O(\log_B m \cdot \log_2 \log_B n + \log_B n)$ I/O span, 
  where $n$ and $m \leq n$ are the sizes of the two trees,
  and $B$ is the block size.
  

\end{abstract}

\newpage

\section{Introduction}
\label{sec:intro}

Balanced search trees are among the most fundamental data structures
in computer science.  The search tree structure effectively maintains
the ordering for a set of elements with dynamic updates, and a
balance guarantee, usually meaning to bound the tree height to be
logarithmic, enables efficient cost bounds for both updates and
queries.  Balanced search trees have been used to support basic
data types such as ordered sets and maps in various programming
languages, either as built-in data types or in standard libraries.

In addition to individual updates such as insertions and deletions, set
operations (e.g., set-union, corresponding to merging two trees)
are often needed in the interface of ordered sets and maps.  To
support efficient set operations on trees, many existing papers
study the \emph{join-based framework}. The core of the framework
is a function \emph{join}. $\funcjoin(T_L,k,T_R)$ takes two search trees
$T_L$ and $T_R$, and one key $k$ in the middle, such that $\forall
k_1\in T_L$ and $k_2\in T_R$, $k_1<k<k_2$, and returns a valid balanced
tree containing all elements in $T_L\cup\{k\}\cup T_R$. Existing work
has shown that abstracting the \funcjoin{} primitive allows for
elegant and efficient designs for set operations on various types
of balanced trees, both
sequentially~\cite{adams1992implementing,adams1993functional} and
in
parallel~\cite{Blelloch1998,blelloch2016just,sun2018pam,blelloch2020optimal,dhulipala2019low,dhulipala2022pac}.

The join-based framework was first studied in the sequential setting, on red-black trees~\cite{Tarjan83} and weight-balanced trees~\cite{adams1992implementing}. The idea then has received particular attention in the parallel setting, since it conceptually allows for applying a batch of updates to the tree. 
Blelloch and Reid-Miller~\cite{Blelloch1998} first used the join-based framework in parallel on treaps, and proved that they are theoretically efficient. 
This idea was extended to various balancing schemes in 2014 on a data structure called P-tree~\cite{blelloch2016just,ptreedb} and later implemented in a parallel library~\cite{sun2018pam}. The theoretical results and implementations of these trees have been used both to support strong theoretical bounds in other parallel algorithms~\cite{cao2023parallel,dong2021efficient,shen2022many,dhulipala2022hierarchical}, 
as well as in various applications such as databases~\cite{ptreedb}, graph processing~\cite{dhulipala2019low,dhulipala2019sage}, and more~\cite{sun2019implementing,wang2023fast}.

Despite the widespread adoption of parallel \join-based trees, a
major drawback of previous work on
such data structures is the inefficiency of their
input/output (I/O) access patterns; e.g., 
see~\cite{akhremtsev16}.
For example,
in a tree
node that only contains a few elements in the set, the number of memory
accesses is asymptotically the same as the operations applied, rather
than being a function of a block size or cache-line size.
In fact, many follow-up works (e.g., C-tree~\cite{dhulipala2019low} and CPAM~\cite{dhulipala2022pac}), aim to improve the memory-access efficiency (i.e., the I/O efficiency) with blocked tree nodes or leaves. 
While these new data structures have provable guarantees in work and span, and have shown good performance in experiments due to better memory access patterns, we are unaware of any theoretical improvements in their I/O efficiency. 
It is therefore worth asking, whether we can adapt the simple \join-based framework to also achieve I/O-efficiency with provable guarantees. 

This paper studies efficient parallel algorithms for set operations 
based on search tree algorithms using a \join-based framework, 
with a special focus on achieving I/O efficiency in these algorithms.
Our new solution is based on B-trees, and includes novel ideas to overcome challenges in both algorithm design and analysis for this problem.
The challenge to achieving I/O-efficiency in parallel is two-fold. 
The first reason, which lies in the algorithm design aspect, is to apply the idea to an I/O-friendly data structure, such as B-trees, when binary structures
are the norm. 
As discussed above, there exist previous studies that attempt to make \join-based algorithms I/O-friendly by grouping multiple keys as a block in the same tree node or in the leaves~\cite{dhulipala2019low,dhulipala2022pac}. 
However, they tend to still keep the binary structure of the tree, which
introduces needless complications. 
This design preserves compatibility with the existing \join-based algorithmic framework, which can be used out of the box with minimal adaptation. However, since the tree height is $O(\log(n/B))$ for block size $B$ and tree size $n$, this design does not lead to an ideal I/O bound for these set algorithms. 
To achieve non-trivial bounds, allowing multi-way trees (e.g., B-trees) is critical. However, the multi-way structure of these structures inevitably introduces complications into algorithm design. In particular, a multi-way \join{} is required, which basically concatenates a set of trees with keys in between. Accordingly, a multi-split algorithm is needed to divide a tree using multiple splitter keys. Both of them need careful algorithm re-design over the existing ideas. 

The second, and perhaps most interesting, reason related to the analysis aspect is the lack of an effective model to capture the I/O efficiency in these algorithms in parallel. 
While both models for parallel computing (e.g., PRAM) and I/O models have a long history, analyzing the I/O cost in parallel has some inherent difficulties such as shared/separate cache, synchronization schemes, exclusive/concurrent writes, etc. 
In 2010, Arge, Goodrich, and Sitchinava
introduce the parallel external-memory (PEM) 
model~\cite{DBLP:conf/ipps/ArgeGS10}, 
which focuses on the I/O bottleneck in PRAM algorithms. 
Built on top of the PRAM model, PEM assumes synchronized threads.
On the other hand, the \join-based tree algorithms, as well as their implementations in software libraries, are based on the classical fork-join model, which is highly asynchronous. 

\subsection*{Our Results} 
To accurately analyze I/O cost in asynchronous parallel algorithms, 
we introduce the \emph{\ourmodel}
in this paper, which more formally defines the I/O cost in fork-join
parallelism. Analogous to work and span for the standard fork-join
model, the \ourmodel{} measures the I/O cost of parallel algorithms by
both the total number of block transfers (referred to as \defn{I/O work}),
and the maximum number of block transfers one depends on the previous
(referred to as \defn{I/O span}). Based on this model, we provide and
analyze new parallel and I/O-efficient algorithms on set/map operations using B-trees. 

Intuitively,
the \ourmodel{} provides a 
software-based asynchronous alternative to the hardware-based 
synchronous parallel external-memory (PEM) model. 
We also present parallel I/O-efficient algorithms for B-trees in this model, such as union, intersection, and difference operations. 
To do this,
we first design a $B$-way \join{} algorithm,
which may concatenate at most $B$ trees. 
Using this primitive, we design a general \join{} algorithm that takes any number of $k$ trees and $k-1$ keys in between, and an inverse function, that splits a tree by $k-1$ keys. 
Based on the two functions, we design an I/O-efficient parallel 
set algorithms.
We also show in the Appendix how to modify the union algorithm to also perform
the intersection and difference operations,
while achieving the same I/O bounds asymptotically.

We highlight that the design of the algorithm is highly non-trivial. Directly applying the simple $B$-way \join{} (and a corresponding $B$-way split) algorithm to a \fname{Union} algorithm leads to $O(\log_B n \log_B m)$ I/O span. 
A specific interesting technical contribution in this paper is to achieve a stronger I/O span bound of $O(\log_B m \cdot \log_2 \log_B n +\log_B n)$, for which we introduce the more sophisticated algorithms of arbitrary-way \join{} 
and split.
We summarize the main results below. 

\begin{restatable}[Parallel Set Operations on B-trees]{thm}{btreemain}
\label{thm:btree-main}
    Given two B-trees with sizes $m$ and $n\ge m$, there exists a parallel algorithm that returns a new B-tree containing the union of the two input trees in and $O\left(m\log_B\left(\frac{n}{m}\right)\right)$ I/O work, $O(\log_B m \cdot \log_2 \log_B n +\log_B n)$ I/O span, where $B$ is the block size. 
\end{restatable}
Due to space limit, we present the \fname{Union} algorithm and analyze it in the main paper, 
and show how to achieve the same bounds for \intersection{} and \difference{} in \cref{sec:additional-set-operations}. 


\section{Preliminaries}
\label{sec:prelim}

\myparagraph{B-trees.}
A B-tree is a multi-way self-balancing tree data structure that maintains an ordered set of keys. 
With clear context, for a B-tree $T$, we also use $T$ to denote the set of keys in $T$. 
We use the standard definitions of parent, child, sibling, ancestor, and descendants in B-trees. 
More formally, a B-tree is either an empty node (external node), 
or a node with a set of $b$ keys $k_1,k_2,\dots k_b$ with $b+1$ subtrees $T_1,T_2,\dots T_{b + 1}$ with the following invariants:
\begin{itemize}
    \item $\lceil B/2\rceil \le b\le B$ for some parameter $B$, except for the root, where $2\le b\le B$.
    \item Each subtree $T_i$ is a B-tree.
    \item $\forall x\in T_{i-1}$ and $y\in T_i$, we have $x<k_i<y$. 
    \item All external nodes have the same depth, where the depth of a node is the number of ancestors (inclusive) of it. 
\end{itemize}

The height of a B-tree $T$, denoted as $h(T)$, is 0 if it is an external node, 
and $h(T')+1$ if $T$ has a child $T'$. If $T$ has multiple children, all of them should have the same height. 

When we discuss an algorithm on two B-trees, we always use $n$ to denote the one with the larger size and $m\le n$ the smaller size. 

The B-tree includes an operation for combining nodes when the number of keys in
the node is below $\lceil B/2\rceil,$
and for splitting a node into two when there are more than $B$ keys in the node.
We will refer to these operations as \defn{fuse} and \defn{divide} operations,
to avoid conflicting terminology.

\myparagraph{Sequence Notation.} We use $a_{i..j}$ ($i\le j$) to denote a sequence of elements $\langle a_i, a_{i+1}, \dots, a_j\rangle$. 

\myparagraph{Background of Join-based Algorithms.} For completeness, we introduce some 
background of \join-based algorithms on binary trees in \cref{app:join-background}. 
We present \cref{fig:framework} as an illustration for the high-level idea, as well as an example of implementing a \algname{Union} based on \join{}. We introduce more details about related work in \cref{app:related}.

One technical challenge to extend this idea to B-trees is to support multi-way \funcjoin{} and \funcsplit{} functions,
which may join multiple trees with keys in the middle, or split a tree using multiple splitters. 
Since the input tree heights may vary, these primitives must rebalance the returned tree properly. 
We introduce our algorithms for these primitives in \cref{sec:multi-split,sec:multi-join}, and finally present our \fastunion{} algorithm in \cref{algo:fastunion}. 

\begin{figure}
    \centering
    \includegraphics[width=\columnwidth]{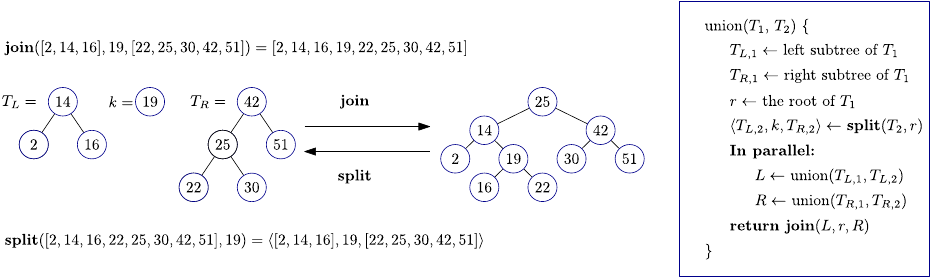}
    \caption{Illustration for \funcjoin{} and \funcsplit{} algorithms on binary trees, and an example of the set \algname{union} algorithm. To extend the idea to B-trees, multi-way \funcjoin{} and \funcsplit{} algorithms are needed.}
    \label{fig:framework}
\end{figure}



\section{The \titlecap{\ourmodel{}}}


In this paper, we propose the \ourmodel{}, which is an extension of the classic work-span model to analyze algorithms. The goal of this paper to use this model is to more precisely capture the I/O costs in algorithms based on fork-join parallelism, which is crucial for analyzing and implementing existing join-based algorithms on trees. 

Recall that in the {work-span model} in the classic multithreaded model
\cite{arora2001thread, blelloch2020optimal, blumofe1999scheduling}, we assume a set of threads share their memory.
Each thread acts like a sequential RAM plus a fork instruction that forks two threads running in parallel.
When both threads are finished with thir task, the original thread continues.
A parallel-for is simulated by forks in a logarithmic number of steps.
A computation can be viewed as a DAG. 
The \textit{work} $W$ of a parallel algorithm is the total number of operations,
and the \textit{span (depth)} $S$ is the longest path in the DAG.
In practice, we can execute the computation with work $W$ and span $S$ using a randomized work-stealing scheduler~\cite{blumofe1999scheduling,gu2022analysis} in time $W/P+O(S)$
with $P$ processors with high probability.

The extension we make to this model to adapt for practicality is that, instead of considering executing one word on RAM, we consider unit cost for reading or writing a block of size $B$, where $B$ is a given parameter.
Here $B$ can be the cacheline size when referring to in-memory algorithms, or the block size when considering external-memory algorithms.
In either case, this is the basic unit for block transfers, which is indeed the motivation of the design of B-trees that improves the performance as compared to a simpler binary search tree.
Here we assume each processor can hold $O(\log n)$ blocks (i.e., the stack memory) to execute the program.
Accessing the data in these blocks is free (it does not asymptotically change the complexity, but only makes the analysis easier).
We define this modified work as the \defn{I/O work}.
Regarding the span, we also adapt similarly (i.e., count the access of an entire block with unit cost), and refer to this as the \defn{I/O span}.
Here we assume a fork or join can create or synchronize $B$ threads for simplicity.
All I/O span bounds remain the same when translating the bound to PRAM, or are affected by a factor of $\log B$ (mostly a small constant) to the binary fork-join model~\cite{blelloch2020optimal}. 

An important subroutine we will be using in this model is the \threadgather{}
primitive.
This operation allows us to gather the $O(B)$ elements which are stored in
various blocks in external memory into a single block in $O(1)$ I/O span and
$O(B)$ I/O work.
During the synchronization of each thread,
the thread may return a $O(1)$ amount of data,
thus the total amount of data returned by all threads is $O(B).$


\section{Parallel Joinable B-Tree}

In this section, we present primitive algorithms for the joinable B-trees. 
We first introduce the \bjoin{} operation,
which takes a sequence of $(b+1)$ B-trees and $b$ keys in between, and
combines them into a valid, balanced B-tree.  
We then show the \bjoin{} algorithm with detailed cost analysis in
\cref{sec:b-way-join}. 
Then, the \bjoin{} algorithm is used as a subroutine to develop the more
general \msplit{} and \mjoin{} algorithms in
\cref{sec:multi-split,sec:multi-join}
in order to achieve better I/O bounds for the \fastunion{} algorithm,
described in \cref{sec:fast-union}.

These join and split algorithms are similar to the approach found in the other
parallel search trees \cite{blelloch2016just, Blelloch1998},
but adapted for the \ourmodel{}.

These algorithms will be used as the building blocks for the set operations
in \cref{sec:fast-union}.
The key idea is to use the \msplit{} operation to split the two B-trees
into $\sqrt{n + m}$ subtrees each,
where the union of a pair of the subtrees is of size $\sqrt{n + m}.$
We then make recursive calls to the union of these pairs,
leading to $\sqrt{n + m}$ total subtrees,
which we then use the \mjoin{} algorithm to combine into a single B-tree.

The \msplit{} operation uses \bjoin{} as a primitive,
by first searching for one of the keys they will be splitting on,
which creates a series of subtrees which must then be joined together,
rebalancing the resulting tree.
The \mjoin{} operation on the other hand is a more complex operation and
differs significantly from the \bjoin{} operation.

\subsection{B-way Join}
\label{sec:b-way-join}

\subsubsection{Algorithm Description}
\begin{algorithm}[t]
\small
\DontPrintSemicolon
\caption{$\bjoin{}(T_{1..(b+1)},k_{1..b})$\label{algo:bjoin}}
\KwIn{A sequence of $b+1$ B-trees $T_{1..(b+1)}$ and $b$ keys $k_{1..b}$. 
For any $x\in T_{i}$ and $y\in T_{i+1}$, we have $x<k_i<y$.}
\KwOut{A B-tree $T$ with all keys in $\{k_1, k_2, \dots, k_b\} \cup \bigcup T_i$}
$\hmax\gets \max_i h(T_i)$\tcp*[f]{$\hmax$ is the largest tree height}\\
$A\gets$ the subsequence of $\langle 1,\dots, b\rangle$, where $i\in A$ iff. $h(T_i)=\hmax$\tcp*[f]{indexes of tall trees}\\
\lIf{the first element in $A$ is not 1}{add $0$ at the beginning of $A$}
\parForEach{$i\in A$}{
$i'\gets$ the successor of $i$ in $A$\tcp*[f]{$T_{i'}$ is the next tree with height $h^*$}\\
\If(\tcp*[f]{Handling the special case when there exist keys to the left of the first tall tree}){$i=0$}{
$R\gets$ the leftmost subtree of $T_{i'}$\\
$S\gets\langle T_{1},\dots T_{i'-1},R\rangle$\tcp*[f]{The sequence of trees for the leftmost recursive call}\\
$k'\gets \langle k_1,k_1,\dots, k_{i'-1}\rangle$ \tcp*[f]{The sequence of keys for the leftmost recursive call}\\
$X\gets\bjoin(S,k')$\tcp*[f]{All keys to the left of the first tall tree}\\
Replace the leftmost subtree of $T_{i'}$ with $X$\tcp*[f]{$T_{i'}$ may be unbalanced for now}\\
}
\Else{
$L\gets$ the rightmost subtree of $T_i$\\
$S\gets \langle L, T_{i+1}, T_{i+2},\dots, T_{i'-1} \rangle$ \tcp*[f]{The sequence of trees for the recursive call}\\
$k'\gets \langle k_{i+1}, k_{i+2},\dots, k_{i'-1} \rangle$\tcp*[f]{The sequence of keys for the recursive call}\\
$X\gets\bjoin(S,k')$\tcp*[f]{$X$ will contain all keys between the $i$-th and the $i'$-th tree}\\
Replace the rightmost subtree of $T_{i}$ with $X$\tcp*[f]{$T_{i}$ may be unbalanced for now}
}
}
\lIf{the first element in $A$ is 0}{remove $0$ from $A$}
$T\gets$ Concatenate all (new) trees $T_i$, with keys $k_{i-1}$, for all $i\in A$\\
\fname{SimpleRebalance}$(T)$ \tcp*[f]{Introduced in \cref{lemma:multiway-join-rebalance-bounds}}
\end{algorithm}

The join operation merges $b + 1$ B-trees $T_1, T_2, \dots, T_{b+1},$ and
$b$ separator keys $k_1, k_2, \dots, k_{b}$,
into a single valid B-tree,
where $b \leq B.$
Each input tree $T_i$ must be a valid B-tree and $T_i$ contains keys only in the
range $(k_{i-1}, k_{i}),$ where $k_0 = -\infty$ and $k_{b + 1} = \infty$.  

The primary challenge lies in ensuring the join operation is both
computationally efficient and maintains the balance of the tree,
even when the height of the input trees may differ significantly.
We describe our algorithm below in three steps. 
An illustration of the three steps on an example input is in \cref{sec:b-way-join}.
In \cref{sec:b-way-join-analysis}, we prove the correctness and efficiency of this algorithm. 

\myparagraph{Step 1: Grouping.} 
The algorithm starts by identifying the height of the tallest tree in the input set,
which we denote as $\hmax.$
We say a tree $T_i$ is a \defn{tall tree} if it has height $\hmax$, and a \defn{short tree} otherwise.
Note that there may be multiple tall trees in the sequence $T$. 
Let $A \subseteq \langle 0, 1, \dots, b\rangle$ be the subsequence of indices of trees with
height $\hmax.$
We then partition the input trees into $|A| + 1$ groups. 
At a high level, each group contains all keys and trees between two tall trees, 
and the leftmost subtree of the previous tall tree. 
If $T_1$ is not a tall tree, a special group, which we call a \defn{leading group}, to the left of $T_1$ will also be formed. 

More formally, assume $i\in A$, and $i'=\successor(i)$ is the successor of $i$ in $A$.\footnote{
  We define $\successor{}(\max(A))$, where $\max(A)$ the last element in $A$, to be $b + 1$.
}
Then the corresponding group includes $i'-i$ trees, which are the rightmost subtree in $T_i$,
and all trees from $T_{i+1}$ to $T_{i'-1}$, and $i'-i-1$ keys, which are from $k_i$ to $k_{i'-2}$. 
Note that the trees are separated by the corresponding keys in the same group. 
If $T_1$ is not a tall tree, all elements to the left of the first tall tree $T_{i'}$ form the leading group, which includes all trees up to $T_{i'-1}$, the leftmost subtree of $T_{i'}$,
and all keys up to $k_{i'-1}$. 

\myparagraph{Step 2: Recursing and Subtree Replacing.} 
After the first step, each group contains a sequence of trees, and the corresponding keys that separate them. The second step simply calls \bjoin{} recursively to reorganize each group into a valid B-tree. After that, for the group corresponding to trees between $T_i$ and $T_{\successor(i)}$,
we use the resulting tree to replace the rightmost subtree in $T_{i}$. 
If there exists a leading group, it replaces the leftmost subtree in the first tall tree.

In \cref{theorem:multiway-join-height} we will show that \bjoin{} on a sequence of trees with maximum height $h$ will lead to a tree with height at most $h+1$. 
In this case, all trees involved in the recursive call is either a short tree, which has height at most $h^*-1$, or a subtree of a tall tree, which has height $h^*-1$. Therefore, the resulting trees from recursive calls have height at most $h^*$. After replacing a subtree in a tall tree, it can be at most one level taller than the other subtrees.

\myparagraph{Step 3: Concatenating and Rebalancing.} 
After the previous step, all short trees have been combined into a tall tree. 
Note that for each tall tree, except for the first one, the key immediately before it is not included in any groups. Therefore, our final step is to concatenate all these remaining components together, which includes all tall trees $T_i$ for all $i\in A$, and $k_{i-1}$ for all $i\in A$. 
This is performed by creating a B-tree node with keys at the root at $T_{A_1}$, then $k_{A_2-1}$, 
then all keys at the root of $T_{A_2}$, then $k_{A_3-1}$, so on so forth. 
The subtrees, from left to right, are all subtrees in $T_{A_1}$, then all subtrees in $T_{A_2}$, so on so forth. Let the resulting tree be $T$. 

The tree $T$ is a valid B-tree except for two aspects: 1) some of its subtrees may be taller than other subtrees, but can be taller by at most 1, and 2) the number of keys at the root may be more than $B$. 
For such minor imbalance, we will use the \fname{SimpleRebalance} algorithm, introduced in \cref{lemma:multiway-join-rebalance-bounds}, to rebalance the tree. 
At a high level, the rebalancing algorithm simply promotes all keys at the subtree root of the taller subtrees to the root of the entire tree, and if an overflow occurs, we can promote every $B$ keys at the root to a higher level, increasing the tree height by 1. 

\definecolor{babyblueeyes}{rgb}{0.63, 0.79, 0.95}
\definecolor{pastelred}{rgb}{1.0, 0.41, 0.38}

\begin{figure}[t]
  \centering
  %
  %
  %
  %
  %
  %
  %
  %
  %
  %
  %
  \includegraphics[width=0.8 \textwidth]{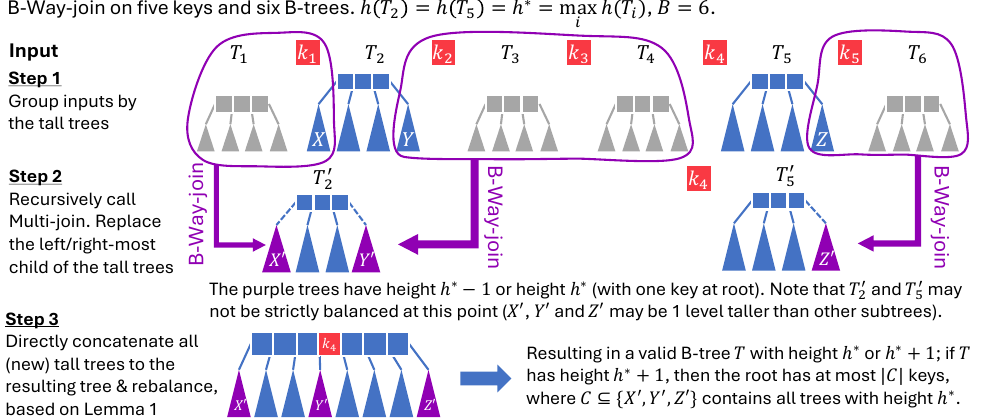}
  \caption{ \small
    An illustration of the \bjoin{} algorithm (described in \cref{sec:b-way-join}) with 6 trees and 5 keys. The blue trees are
    the \textit{tall trees} with height $h^*$. All other trees has height less than $h^*$. 
    \hide{Our algorithm will partition them into 3 groups based on the tall trees. 
    Each group is passed into a recursive call of the MultiwayJoin algorithm,
    and the results will replace the leftmost/rightmost subtrees of the tall trees. 
    Finally, all resulting trees from recursive calls will be concatenated together. 
    When overflow at the root occurs, we rebalance the tree and achieve an output tree with height at most $h^*+1$.}
    \hide{$T_1$ and $C_1$ form the first group, with $k_1$ as the key.
    $C_2, T_3,$ and $T_4$ form the second group, with $k_2, k_3, k_4$ as the keys.
    And $C_3$ and $T_6$ form the third group, with $k_5$ as the key.
    Each group is passed into a recursive call of the MultiwayJoin algorithm.
    The output of the first two recursive calls are made to be the left most
    and right most children of the first tallest tree.
    The remaining outputs are the right most child of their respective tallest
    trees. 
    \yihan{What happens to $k_4$? Will it become an extra key to the second recursive call, since the leftmost subtree in $T_5$ will not be considered in the recursive call? Should we leave it for later joining $T_2$ and $T_5$?}
    }  }
  \label{fig:multiway-join-partitioning}
\end{figure}

See Figure \ref{fig:multiway-join-partitioning} for an illustration of the
algorithm.

\begin{restatable}[B-Way Join Analysis]{thm}{multijoincost}
  Let $T_1, T_2, \dots, T_{b + 1}$ be a set of B-trees,
  with the largest tree height $h_{\max}$ and
  the shortest tree height $h_{\min}$,
  and $k_1, k_2, \dots, k_{b}$ be a set of separator keys,
  where $b \leq B$.
  The \bjoin{} operation can be performed in $O(h_{\max} - h_{\min})$ I/O span and
  $O(b \cdot (h_{\max} - h_{\min}))$ I/O work.
\end{restatable}

We analyze the cost of the \bjoin{} operation in \cref{sec:b-way-join-analysis}.

\subsection{Multi-Split}
\label{sec:multi-split}

The $\msplit{}$ algorithm takes a B-tree $T$ and a sorted list of $d$ keys
$k_{1..d},$
and returns $d + 1$ subtrees,
such that the $i$-th subtree contains all keys in the range $(k_{i - 1}, k_i),$
where $k_0 = -\infty$ and $k_{d + 1} = \infty.$ 
This is achieved by spawning $d + 1$ threads,
where each thread $i$ invokes a helper function $\tsplit{},$
which outputs a single subtree containing all keys in the range
$(k_{i - 1}, k_i).$
See Algorithm \ref{algo:tsplit}.

\begin{algorithm}[t]
\small
\caption{$\tsplit{}(T, k_1, k_2)$\label{algo:tsplit}}
\DontPrintSemicolon
\KwIn{A B-trees $T$ and two keys $k_1, k_2$ such that $k_1 < k_2$}
\KwOut{A B-tree $T$ containing all keys in the range $(k_1, k_2)$}
\lIf{$T$ is empty or contains only keys between $k_1$ and $k_2$}{
  \Return{$T$}
}
\lIf{$T$ is a leaf}{
  \Return{the B-tree containing all keys in the range $(k_1, k_2)$}
}
$r_{1..b} \gets$ the keys at the root of $T$\\
$i \gets$ the smallest index such that $k_1 < r_i,$ 0 if $k_1 < r_1$\\
$j \gets$ the largest index such that $r_j < k_2,$ $b + 1$ if $r_b < k_2$ \\
$C_{i..j} \gets$ all children of $T$ from $i$ to $j$, inclusive\\
\lIf{$i = j$}{
  \Return{$\tsplit{}(C_i, k_1, k_2)$}
}
$C^\prime_i \gets \tsplit{}(C_i, k_1, \infty)$\\
$C^\prime_j \gets \tsplit{}(C_j, -\infty, k_2)$\\
$r^\prime_{i..j} \gets$ copy of keys $r_{i..j}$\tcp*[f]{Prevents adjacent threads from dividing the same node}\\
\Return{$\bjoin{}(\langle C^\prime_i, C_{i + 1..j - 1}, C^\prime_j \rangle, r^\prime_{i..j})$}
\end{algorithm}

The high-level idea of the algorithm is to search for the keys $k_i$ and
$k_{i + 1}$ in $T,$
copying the nodes and keys along the search path that fit within the range
$(k_i, k_{i + 1}).$
Any node that lies on the shared prefix of both search paths,
i.e. before the paths to $k_{i}$ and $k_{i + 1}$ diverge,
are not part of the output tree and can be skipped.
Once the paths diverge,
recursively split down the search path of $k_1$ and $k_2.$
If our entire node falls within the range $(k_1, k_2),$
then a copy of the pointer to the node is made.
Otherwise,
only the keys which fall within the range $(k_1, k_2)$
and the pointers to the children that fall within the range are copied.

We copy the keys of the node instead of dividing it,
which is the natural operation to perform,
because dividing requires synchronization between adjacent threads.
See Figure \ref{fig:tsplit}.

On the way back up the tree,
the algorithm performs a $\bjoin{}$ operation to combine the subtrees and keys
which contain the keys in the range $(k_1, k_2).$
This serves the purpose of both combining and rebalancing the subtrees.
We prove the following theorem in \cref{sec:multiway-split-cost}.

\begin{figure}[t]
  \centering
  \includegraphics[width=0.5 \textwidth]{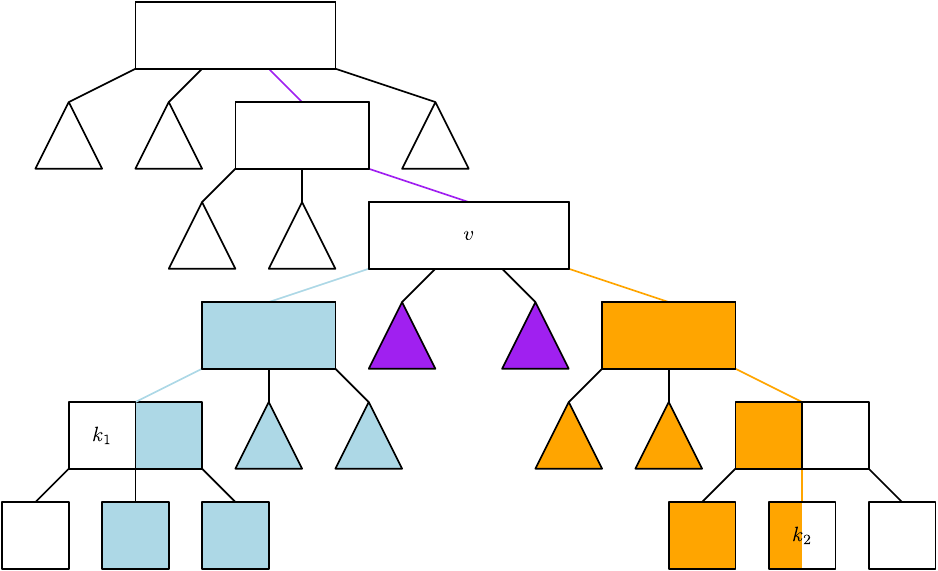}
  \caption{
    Above is a B-Tree, where each node rectangle is a node in the B-Tree,
    whereas the triangles are subtrees shortened for space.
    There are two search paths to $k_1$ and $k_2$ from the root,
    sharing the purple path at the start,
    then separating into the blue and orange paths respectively at node $v.$
    All nodes colored blue are the result of following the blue path and
    copying the appropriate keys,
    and similarly for the orange nodes.
    The node containing $k_2$ excludes all keys greater than or equal to $k_2,$
    thus being half orange.
    The result of $\tsplit{}$ will be the union of the blue, orange, and purple
    nodes.
    All white sections of the tree will be the search paths of other
    \tsplit{} operations.
  }
  \label{fig:tsplit}
\end{figure}

\begin{restatable}[Multiway Split Cost]{theorem}{multiSplitCost}
  \label{thm:multi-split-cost}
  We can split a B-tree $T$ by $k_1, \dots, k_d$ keys in parallel with
  $O(\log_B d + \log_B n)$ I/O span and $O(d \log_B n)$ I/O work,
  where $n$ is $\sum_{i = 1}^d |T_i|,$
  in the \ourmodel{}.
\end{restatable}

\subsection{Multi-Join}
\label{sec:multi-join}

Suppose we have $d + 1$ B-trees $T_{1, \dots, d + 1},$
and a sorted list of $d$ keys $k_{1, \dots, d},$
such that for any $x \in T_i$ and $y \in T_{i + 1},$ we have $x < k_i < y.$
The $\mjoin{}$ algorithm takes as input $T_{1, \dots, d + 1}$ and
$k_{1, \dots, d},$
and returns a single tree $T$,
such that $T$ contains $\bigcup_{i = 1}^d T_i$ and $\{k_1, k_2, \dots, k_d\}.$

We first describe here an overview of the $\mjoin{}$ algorithm.
\cref{sec:b-way-join} shows how to join $B$ trees in
$O(\log_B n)$ I/O span and $O(B \log_B n)$ I/O work.
Therefore,
a simple solution is to join $B$ trees together in parallel across
$\log_B d$ rounds of joining.
However, this results in an algorithm with $O(\log_B d \log_B n)$ I/O span,
which is not optimal for larger values of $d.$

We improve the upper bound on the I/O span by reducing the cost of joining $B$
trees down to an $O(\log_B \log_B n)$ I/O span and $O(\log_B n)$
I/O work,
which we call the $\bjoinfast{}$ operation.
Let us define nodes containing more than $B$ keys as \defn{overloaded},
and nodes with exactly $B$ keys as \defn{full}.
In order to achieve these bounds,
we use a similar strategy to Akhremtsev and Sanders
\cite{DBLP:journals/corr/AkhremtsevS15},
where the trees are preprocessed such that each tree $T_i$ contains a list of
pointers to all nodes along their left and right spines.
This allows us to efficiently find the level at which each tree must traverse in
order to perform the join operation.
For each $\bjoinfast{}$ operation,
the output tree must also maintain this list of pointers.
Then we show how to divide all nodes that will be overloaded during
$\bjoinfast{}$ in parallel.

\subsubsection{Faster B-way Join}
\label{sec:fast-b-way-join}
\begin{algorithm}[t]
\small
\caption{$\bjoinfast{}(T_{1..b+1},k_{1..b})$\label{algo:bjoinfast}}
\KwIn{
  A sequence of $b+1$ B-trees $T_{1..b+1}$ and $b$ keys $k_{1..b}$. 
  For any $x\in T_{i}$ and $y\in T_{i+1}$, we have $x<k_i<y$. 
  For all $T_i,$ there is at most one node with $B$ or more keys along their left and right spines.
}
\KwOut{A B-tree $T$ with all keys in $\{k_1, k_2, \dots, k_b\} \cup \bigcup T_i$}
$\hmax\gets \max_i h(T_i)$\tcp*[f]{$\hmax$ is the largest tree height}\\
$A\gets$ the subsequence of $\langle 1,\dots, b\rangle$, where $i\in A$ iff. $h(T_i)=\hmax$\tcp*[f]{indexes of tall trees}\\
\If{the first key in $A$ is not 1}{
  $i^\prime \gets$ the first key in $A$\\
  $T_{i^\prime} \gets \joinwithtall{}(T_{1..i^\prime}, k_{1..i^\prime - 1})$\\
}
\parForEach{$i\in A$}{
  $i'\gets$ the successor of $i$ in $A$\tcp*[f]{$T_{i'}$ is the next tree with height $h^*$}\\
  $T_{i^\prime} \gets \joinwithtall{}(T_{i..i^\prime}, k_{i..i^\prime - 1}) $
}
\lIf{the first key in $A$ is 0}{remove $0$ from $A$}
$T\gets$ Concatenate all (new) trees $T_i$, with keys $k_{i-1}$, for all $i\in A$\\
\fname{SimpleRebalance}$(T)$ \tcp*[f]{See \cref{lemma:multiway-join-rebalance-bounds}}
\end{algorithm}

Let $L_i$ and $R_i$ be the list of pointers to the left and right spines of
some tree $T_i$ respectively.
$L_i$ and $R_i$ are sorted in reverse order by the height of the node,
i.e. the leaf node is at index 0 and the root is at index $h(T_i).$
Implementation details are discussed in \cref{sec:maintain-pointers}.

\myparagraph{Step 1: Grouping}
As in the \bjoin{} algorithm in \cref{sec:b-way-join},
we first identify the tallest trees of $T_{1, \dots, b},$
and group the trees and keys together.

\myparagraph{Step 2: Fusing Within Groups}
Then, we use the list of pointers to the left and right spines in order to jump
directly to the correct height of the tree for each group.

Each tree $T_j$ is assigned some $T_i$ to be fused with,
where $i$ is the largest index within the same group such that $i < j$ and
$h(T_i) \geq h(T_j).$
Then a thread is spawned for each $T_j$,
and its corresponding key are fused with $L_i[h(T_j)]$ or $R_i[h(T_j)].$
As in Section \ref{sec:b-way-join},
we also account for the special case where $T_1$ is not the tallest tree.

\begin{figure}[t]
  \centering
  \includegraphics[width=0.8 \textwidth]{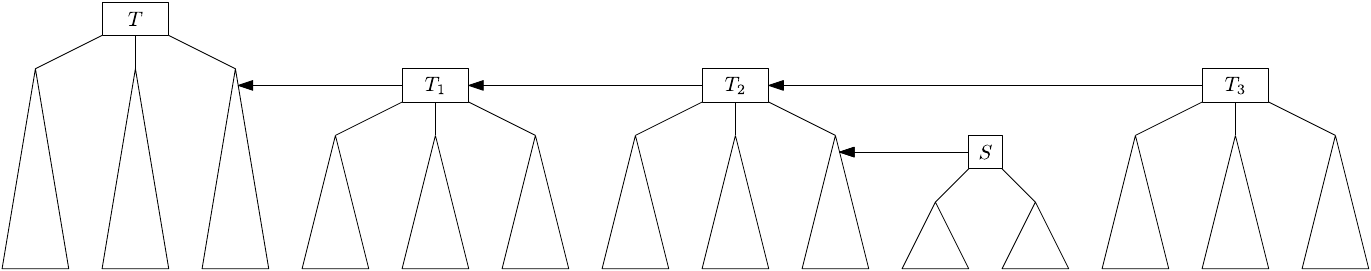}
  \caption{
    Trees $T_1, T_2,$ and $T_3$ are all trees of the same height in the same
    group.
    $T_1$ fuses with a node in $T,$ but at the same time,
    $T_2$ is also fusing with the root node of $T_1,$
    and $T_3$ is fusing with the root node of $T_2.$
    Therefore,
    there is a dependency where $X$ must fuse with $T_2$ first,
    then $T_2$ must fuse with $T_1,$
    then $T_1$ with $T.$
    While $S$ is also fusing with a node in $T_2,$
    $S$ is fusing with node other than the root of $T_2,$
    thus there is no conflict.
  }
  \label{fig:bfast-conflict}
\end{figure}

Conflicts occur when multiple trees of the same height are in the same group, 
since threads may try to fuse with roots of trees that are also
fusing with other trees.
However,
in constant I/O span and work,
each thread is able to identify which portion of the tree they will end up
fusing with and directly write to that node.
Consider the tree $T_i,$
and let $\treejoinset{}$ be the set of all trees with the same height as $T_i.$
$\treejoinset{}$ can be constructed using \threadgather{} to get the heights of
all $T_i$ in a single cache line and constructing each set in internal memory.
Then all $B$ threads are spawned,
and thread $i$ can check if $|\treejoinset{}| > 1.$
If so,
thread $i$ can also identify the tree in $\treejoinset{}$ with the smallest
index and the tree $T^\prime$ they would be fusing with.
If $T^\prime$ exists,
then thread $i$ can fuse $T_i$ with $T^\prime$ instead,
while ensuring there are no write conflicts with other threads which are
fusing with the same node.
For example,
in Figure \ref{fig:bfast-conflict},
$T_1$ is writing keys to some node in $T,$
which we call $v.$
$T_2$ will also write to $v,$
but starting from the index $|v| + |T_1|$ away from the start of $|v|.$

\myparagraph{Step 3: Dividing Large Nodes}
After each thread has finished fusing their trees,
they check if the node is overloaded.
If so,
then the thread will divide the node into $\lceil \frac{b}{B} \rceil$ nodes,
where $b$ is the number of keys in the node.
However,
before pushing $\lceil \frac{b}{B} \rceil - 1$ keys up to the parent node,
all threads are first synchronized in order to avoid write conflicts.
Then all threads are spawned again to push up their keys to the parent node.
We show in \cref{lem:multiway-join-divide-up-bound},
in \cref{sec:multi-join-cost-analysis},
all overloaded nodes will push at most one key to their parent node.

\myparagraph{Step 4: Dividing Smaller Nodes in Parallel}
Any remaining overloaded nodes will only cause consecutive full nodes above it
to be divided,
since a single key is pushed up per overloaded node.
By computing the number of consecutive full nodes above each overloaded node,
we can divide these nodes in parallel.

Using the prefix sum algorithm,
described in \cref{sec:prefix-sum} of the Appendix,
we can accomplish this.
Create an array $C_v$ of length $\log_B n$ for each overloaded node $v,$
where $C_v[i]$ is set to 1 if the node $i$ levels above is full or overloaded,
and 0 otherwise.
If there is no node $i$ levels above $v$,
then $C_v[i]$ is set to 0.
Then take the prefix sum of each $C_v,$
and mark the node $i$ levels above $v$ to be divided if $C_v[i] = i.$
The number of threads which marks a node to be divided is the number of
keys pushed up into it.
Finally,
a thread is spawned for each node which is marked to be divided,
and is divided in parallel.

\begin{figure}[t]
  \centering
  \includegraphics[width=0.8 \textwidth]{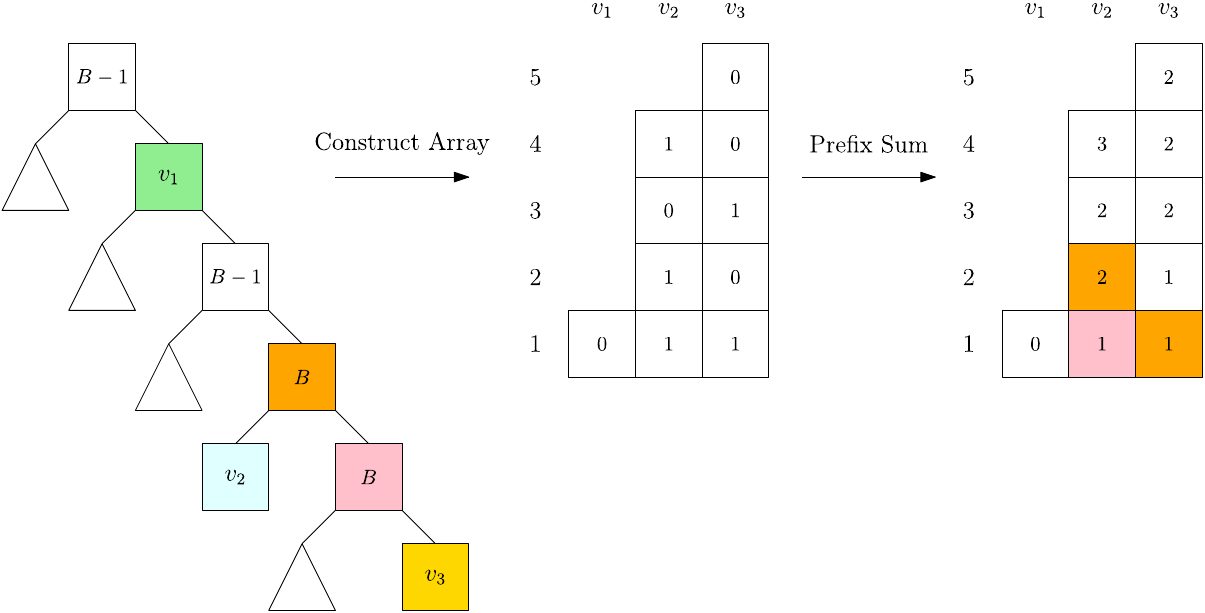}
  \caption{
    The B-tree contains 3 nodes $v_1, v_2,$ and $v_3$ which are overloaded.
    We use the list of pointers in order to construct our arrays $C$.
    A cell within the array is set to 1 if the node is full or overloaded,
    i.e. it contains $B$ or more keys.
    Otherwise, it is set to 0.
    Note that since $v_1$ only contains a single node above it,
    its array is also of length 1.
    Then we take the prefix sum of each array,
    and mark a node to be divided if $C[i] = i.$
    The three colored cells are the nodes which are to be divided.
    Two of the cells are orange,
    referring to the same node,
    meaning the node takes a key from two of its children.
  }
  \label{fig:bfast-divide-nodes}
\end{figure}

\myparagraph{Step 5: Concatenating and Rebalancing}
Just as in the \bjoin{} algorithm,
after the previous steps,
all short trees have been combined into a single tall tree,
which is then rebalanced.

\subsubsection{Maintaining the Left and Right Pointers}
\label{sec:maintain-pointers}
In order to achieve the I/O span and work bounds we desire,
maintaining the list of pointers along the spines of the trees is crucial.
We achieve this by implementing the lists as sequential B-trees,
where the keys are pointers to the nodes along the spine. 
By using order statistics information,
searching for the node at height $i$ takes $O(\log_B \log_B n)$ I/O span and work.

If a new node is created for the left or right spine after the dividing process,
we add the node to the front of $L_i$ or $R_i$ respectively.
This also takes $O(\log_B \log_B n)$ I/O span and work.
Finally,
after the fusing operations complete,
we must also update the pointers to the left/right spines of the trees.
Recall in \cref{fig:bfast-conflict} that part of the right spine of $T_1$
is being covered by the tree $T_3,$
specifically all nodes with a height lower than $T_3$'s root are not on
the right spine of the output tree.
More generally,
suppose $S$ is the set of all trees that were grouped together and has at least
one node from their right spine which is part of the right spine of the output
tree.
Then $T \in S$ will have all nodes along its right spine less than the height
of the next smaller tree in $S$ covered.
These are the nodes that should be removed from their list of spines.
To do this,
we simply find the height, $h,$ of the next smaller tree in $S$ and perform a
sequential split on $T.$
Then we use \bjoin{} to join all trees together.

\section{Union}
\label{sec:fast-union}

\subsection{Algorithm Description}
Let $T_1$ and $T_2$ be two B-trees with $n$ and $m$ keys respectively.

The union algorithm is a simple $d$-way divide and conquer algorithm,
where $d = \sqrt{n + m}.$
If $m < \sqrt{n + m},$
then we have reached the base case,
and we call a I/O work inefficient union algorithm.
Let $|T_1| \geq |T_2|,$
as we can swap the two trees if this is not the case.
The I/O work inefficient union algorithm uses all the keys in $T_2$ and uses them as
keys to split on $T_1$ using the $\msplit{}$ operation.
Then we use the $\mjoin{}$ operation to combine the $d$ subtrees and the $d - 1$
keys we split on.

Otherwise,
we search $T_1$ and $T_2$ for $d - 1$ keys,
which we will use to split each tree into $d$ subtrees using the $\msplit{}$
operation.
We denote the subtrees of $T_1$ as $T_{1, \{1, \dots, d \}}$
and $T_2$ as $T_{2, \{1, \dots, d\}}$.
The $i$-th key which we split on is the global $(n + m) \cdot i / d)$-th key.
This ensures that $|T_{1, i}| + |T_{2, i}|$ is all approximately
$\frac{n + m}{d}$.
We then spawn $O(d)$ threads,
and recursively call the union algorithm on each pair of subtrees
$T_{1, i}$ and $T_{2, i}$.
Then we use the $\mjoin{}$ operation to combine the $d$ subtrees and the
$d - 1$ keys we split on.

\begin{algorithm}[t]
\small
\caption{$\fastunion{}(T_1, T_2)$\label{algo:fastunion}}
\KwIn{Two B-trees $T_1$ and $T_2$}
\KwOut{A B-tree $T$ with all keys in $T_1 \cup T_2$}
\lIf{$|T_2| > |T_1|$}{
  $T_1, T_2 \gets T_2, T_1$
}
\lIf{$|T_2| < \sqrt{|T_1| + |T_2|}$}{
  \Return{$\union{}(T_1, T_2)$}
}
$d \gets \sqrt{|T_1| + |T_2|}$\\
\parForEach{$1 \leq i \leq d$}{
  $k_i \gets \twobinarysearch{}(\frac{(|T_1| + |T_2|) i}{d}, T_1, T_2)$
}
$\langle T_{1, 1}, T_{1, 2}, \dots, T_{1, d + 1} \rangle \gets \msplit{}(T_1, k_{1..d})$\\
$\langle T_{2, 1}, T_{2, 2}, \dots, T_{2, d + 1} \rangle \gets \msplit{}(T_2, k_{1..d})$\\
\parForEach{$1 \leq i \leq d + 1$}{
  $T_i \gets \fastunion{}(T_{1, i}, T_{2, i})$\\
}
\Return{\mjoin{}($T_{1..d + 1}, k_{1..d}$)}
\end{algorithm}

\subsection{Cost Analysis}

Due to space constraints,
we have put the details of the analysis for the base case of the union algorithm
to \cref{sec:union-base-case-cost}.

\begin{restatable}[Union Algorithm Base Case Cost]{lemma}{unionBaseCaseCost}
  \label{lem:union-base-case-cost}
  The \union{} algorithm takes $O(\log_B m \cdot \log_2 \log_B n + \log_B n)$
  I/O span and $O(m \log_B n)$ I/O work.
\end{restatable}

\begin{lemma}
  Algorithm \ref{algo:fastunion} takes
  $O(\log_B m \cdot \log_2 \log_B n)$ I/O span.
\end{lemma}

\begin{proof}
  The I/O span of finding $d - 1,$
  where $d = \sqrt{n + m},$
  keys and splitting along them costs
  $O(\log_B (\sqrt{n + m})),$
  and the I/O span of joining all the recursive components costs
  $O(\log_B (\sqrt{n + m}) \cdot \log_2 \log_B n).$
  Recall that $m > \sqrt{n + m},$
  since otherwise we would be in the base case of the algorithm.
  Therefore,
  the total I/O span of the above operations is $O(\log_B m \cdot \log_2 \log_B n).$

  Then for the next recursive call,
  the problem size shrinks to $\sqrt{n + m}.$
  The I/O span of the base case is also $O(\log_B m \cdot \log_2 \log_B n).$
  Therefore,
  the cost of the I/O span is upper bounded by

  $$
    S(n + m) = S(\sqrt{n + m}) + O(\log_B m \cdot \log_2 \log_B (n + m)).
  $$

  Solving for this recurrence relation gives us
  $O(\log_B m \cdot \log_2 \log_B n),$
  when $m < n.$
\end{proof}

\begin{lemma}
  The total I/O work of the \fastunion{} algorithm is $O(m \log_B \frac{n}{m}).$
\end{lemma}

\begin{proof}
  The total cost of the base case for the \fastunion{} algorithm is
  $O(m \log_B \frac{n}{m})$ from Lemma \ref{lem:union-base-case-cost} in the
  Appendix.
  Thus,
  we focus on the cost of the recursive calls.

  We see that the I/O work cost is dominated by the cost of \msplit{}, \mjoin{},
  and searching for the split keys.
  All of these operations take $O(\sqrt{n + m}\log_B (n + m))$ I/O work.

  We perform a similar analysis to the one in Blelloch et al.
  \cite{blelloch2020optimal}.
  Consider the recurrence tree of this algorithm,
  which is a $\sqrt{n + m}$-way tree,
  and the cost of the current node is $\sqrt{n + m} \log_B (n + m).$
  Since the recursion stops when $m < \sqrt{n + m},$
  the recurrence tree is not balanced and some branches will reach the base
  case earlier than others.
  However,
  the cost of a node at the same level also costs the same,
  and decreases as the tree goes deeper.
  Let $t = n + m,$
  then the total I/O work of \fastunion{} is:

  $$  t^{1 / 2} \log_B t + t^{1 / 2} \cdot t^{1 / 4} \log_B t^{1 / 2} + \dots = \sum_i \frac{\log_B t}{2^i} t^{1 - 1 / 2^{i + 1}}. $$

  We can see the above recurrence is leaf dominated.
  For all $i,$ $t^{1 - 1 / 2^{i + 1}}$ is the number of split keys which are
  being searched.
  Naturally the total number of split keys can never be more than $m,$
  else we would have already reached the base case,
  meaning that $t^{1 - 1 / 2^{i + 1}} \leq m.$
  Let $x = t^{1 / 2^{i + 1}},$
  then the total leaf cost is

  $$
    O(t^{1 - 1 / 2^{i + 1}} \cdot \log_B t^{1 / 2^{i + 1}}) = O\left(\frac{t}{x} \log_B x\right).
  $$

  Clearly the above function decreases as $x$ approaches infinity,
  and since $t^{1 - 1 / 2^{i + 1}} \leq m,$
  we know that $x \geq t / m.$
  Letting $x = t / m$ gives us the maximum value of the function,
  which is $O(m \log_B (\frac{n}{m})),$
  giving us our desired I/O work bounds.
\end{proof}

Combining the results above we have the main theorem of this paper. 

\btreemain*

We show in \cref{sec:additional-set-operations} how to achieve
the above bounds for the other set operations.

\section{Conclusion}
In this paper,
we have introduced the parallel set operations on B-trees that achieve provable
I/O efficiency in the context of a novel cost model, 
the \ourmodel{},
which offers a way to analyze the I/O costs of algorithms in the fork-join
setting.
This captures the I/O complexity of algorithms executed under asynchronous
fork-join parallelism,
unlike the Parallel External Memory (PEM) model,
which assumes synchronized threads at each step of the algorithm.
Our model provides a more realistic abstraction for analyzing parallel
algorithms that rely on asynchronous scheduling.
We showed how to adapt the join-based paradigm to B-trees by developing
new algorithms for \mjoin{}, and \msplit{}.
These primitives serve as building blocks for implementing more complex
operations,
such as the union, intersection, and difference operations.
This allows us to achieve a near-optimal I/O work of $O\left(m\log_B\left(\frac{n}{m}\right)\right)$,
and I/O span of $O(\log_B m \cdot \log_2 \log_B n +\log_B n)$.
Potential future work would be to reduce the I/O work to the optimal
$O\left(\frac{m}{B}\log_B\left(\frac{n}{m}\right)\right)$ I/O complexity,
and reduce the I/O span to $O\left(\log_B (n + m)\right).$

\clearpage
\bibliography{main}

\newcommand{\noopsort}[1]{} \newcommand{\singleletter}[1]{#1}
\begin{thebibliography}{10}

\bibitem{adams1992implementing}
Stephen Adams.
\newblock Implementing sets effciently in a functional language.
\newblock Technical Report CSTR 92-10, University of Southampton, 1992.

\bibitem{adams1993functional}
Stephen Adams.
\newblock Efficient sets---a balancing act.
\newblock {\em Journal of functional programming}, 3(04):553--561, 1993.

\bibitem{DBLP:journals/corr/AkhremtsevS15}
Yaroslav Akhremtsev and Peter Sanders.
\newblock Fast parallel operations on search trees.
\newblock {\em CoRR}, abs/1510.05433, 2015.
\newblock URL: \url{http://arxiv.org/abs/1510.05433}, \href {https://arxiv.org/abs/1510.05433} {\path{arXiv:1510.05433}}.

\bibitem{akhremtsev16}
Yaroslav Akhremtsev and Peter Sanders.
\newblock Fast parallel operations on search trees.
\newblock In {\em {IEEE} International Conference on High Performance Computing (HiPC)}, 2016.

\bibitem{DBLP:conf/ipps/ArgeGS10}
Lars Arge, Michael~T. Goodrich, and Nodari Sitchinava.
\newblock Parallel external memory graph algorithms.
\newblock In {\em {IPDPS}}, pages 1--11. {IEEE}, 2010.

\bibitem{arora2001thread}
Nimar~S. Arora, Robert~D. Blumofe, and C.~Greg Plaxton.
\newblock Thread scheduling for multiprogrammed multiprocessors.
\newblock {\em Theory of Computing Systems}, 34(2):115--144, 2001.

\bibitem{blelloch2016just}
Guy~E. Blelloch, Daniel Ferizovic, and Yihan Sun.
\newblock Just join for parallel ordered sets.
\newblock In {\em {ACM} Symposium on Parallelism in Algorithms and Architectures (SPAA)}, 2016.

\bibitem{blelloch2019optimal}
Guy~E. Blelloch, Jeremy~T. Fineman, Yan Gu, and Yihan Sun.
\newblock Optimal parallel algorithms in the binary-forking model.
\newblock {\em arXiv preprint 1903.04650}, 2019.

\bibitem{blelloch2020optimal}
Guy~E. Blelloch, Jeremy~T. Fineman, Yan Gu, and Yihan Sun.
\newblock Optimal parallel algorithms in the binary-forking model.
\newblock In {\em {ACM} Symposium on Parallelism in Algorithms and Architectures (SPAA)}, pages 89--102, 2020.

\bibitem{Blelloch1998}
Guy~E. Blelloch and Margaret Reid-Miller.
\newblock Fast set operations using treaps.
\newblock In {\em {ACM} Symposium on Parallelism in Algorithms and Architectures (SPAA)}, pages 16--26, 1998.

\bibitem{blumofe1999scheduling}
Robert~D. Blumofe and Charles~E. Leiserson.
\newblock Scheduling multithreaded computations by work stealing.
\newblock {\em Journal of the ACM}, 46(5):720--748, 1999.

\bibitem{cao2023parallel}
Nairen Cao and Jeremy~T Fineman.
\newblock Parallel exact shortest paths in almost linear work and square root depth.
\newblock In {\em Proceedings of the 2023 Annual ACM-SIAM Symposium on Discrete Algorithms (SODA)}, pages 4354--4372. SIAM, 2023.

\bibitem{dhulipala2022pac}
Laxman Dhulipala, Guy~E. Blelloch, Yan Gu, and Yihan Sun.
\newblock {PaC-trees}: Supporting parallel and compressed purely-functional collections.
\newblock In {\em ACM Conference on Programming Language Design and Implementation (PLDI)}, 2022.

\bibitem{dhulipala2019low}
Laxman Dhulipala, Guy~E Blelloch, and Julian Shun.
\newblock Low-latency graph streaming using compressed purely-functional trees.
\newblock In {\em ACM Conference on Programming Language Design and Implementation (PLDI)}, pages 918--934, 2019.

\bibitem{dhulipala2022hierarchical}
Laxman Dhulipala, David Eisenstat, Jakub Lacki, Vahab Mirrokni, and Jessica Shi.
\newblock Hierarchical agglomerative graph clustering in poly-logarithmic depth.
\newblock {\em arXiv preprint:2206.11654}, 2022.

\bibitem{dhulipala2019sage}
Laxman Dhulipala, Charlie McGuffey, Hongbo Kang, Yan Gu, Guy~E Blelloch, Phillip~B Gibbons, and Julian Shun.
\newblock Sage: Parallel semi-asymmetric graph algorithms for nvrams.
\newblock {\em arXiv preprint arXiv:1910.12310}, 2019.

\bibitem{dong2021efficient}
Xiaojun Dong, Yan Gu, Yihan Sun, and Yunming Zhang.
\newblock Efficient stepping algorithms and implementations for parallel shortest paths.
\newblock In {\em {ACM} Symposium on Parallelism in Algorithms and Architectures (SPAA)}, pages 184--197, 2021.

\bibitem{gu2022analysis}
Yan Gu, Zachary Napier, and Yihan Sun.
\newblock Analysis of work-stealing and parallel cache complexity.
\newblock In {\em {SIAM} Symposium on Algorithmic Principles of Computer Systems (APOCS)}, pages 46--60. SIAM, 2022.

\bibitem{marlow2010haskell}
Simon Marlow et~al.
\newblock Haskell 2010 language report.
\newblock {\em Available online http://www. haskell. org/(May 2011)}, 2010.

\bibitem{Leda99}
Kurt Mehlhorn and Stefan N\"{a}her.
\newblock {\em LEDA: A Platform for Combinatorial and Geometric Computing}.
\newblock Cambridge University Press, New York, NY, USA, 1999.

\bibitem{shen2022many}
Zheqi Shen, Zijin Wan, Yan Gu, and Yihan Sun.
\newblock Many sequential iterative algorithms can be parallel and (nearly) work-efficient.
\newblock In {\em {ACM} Symposium on Parallelism in Algorithms and Architectures (SPAA)}, 2022.

\bibitem{sun2019implementing}
Yihan Sun and Guy Blelloch.
\newblock Implementing parallel and concurrent tree structures.
\newblock In {\em {ACM} Symposium on Principles and Practice of Parallel Programming (PPOPP)}, page 447–450, 2019.

\bibitem{ptreedb}
Yihan Sun, Guy~E Blelloch, Wan~Shen Lim, and Andrew Pavlo.
\newblock On supporting efficient snapshot isolation for hybrid workloads with multi-versioned indexes.
\newblock {\em Proceedings of the VLDB Endowment (PVLDB)}, 13(2):211--225, 2019.

\bibitem{sun2018pam}
Yihan Sun, Daniel Ferizovic, and Guy~E Blelloch.
\newblock {PAM}: Parallel augmented maps.
\newblock In {\em {ACM} Symposium on Principles and Practice of Parallel Programming (PPOPP)}, 2018.

\bibitem{Tarjan83}
Robert~Endre Tarjan.
\newblock {\em Data Structures and Network Algorithms}.
\newblock Society for Industrial and Applied Mathematics, Philadelphia, PA, USA, 1983.

\bibitem{wang2023fast}
Letong Wang, Xiangyun Ding, Yan Gu, and Yihan Sun.
\newblock Fast and space-efficient parallel algorithms for influence maximization.
\newblock {\em Proceedings of the VLDB Endowment (PVLDB)}, 17(3), 2023.

\end{thebibliography}

\clearpage
\appendix
\section{Background for \join{}-based Algorithms on Binary Trees}
\label{app:join-background}

For completeness, we briefly introduce the background of join-based algorithms on binary trees. For binary trees, the $\funcjoin(L,k,R)$ function takes two search trees $T_L$ and $T_R$, and one key $k$ in the middle, such that $\forall
k_1\in T_L$ and $k_2\in T_R$, $k_1<k<k_2$, and returns a valid balanced
tree containing all elements in $L\cup\{k\}\cup R$. 
As the inversion to \funcjoin{}, a \funcsplit{}$(T,k)$ is defined to return $\langle T_L,k',T_R \rangle$,
where $T_L$ ($T_R$) contains all keys in $T$ smaller (larger) than $k$, and $k'$ is the node containing $k$ if $k\in T$, and null otherwise. 
In this case, an example of the \algname{Union} algorithm can be implemented as illustrated in \cref{fig:framework}. 
Note that the \algname{Union} algorithm is independent of the balancing schemes used, as long as the \funcjoin{} and \funcsplit{} functions are properly supported. 

\section{Related Work}
\label{app:related}

The \funcjoin{} operation was first proposed by Tarjan~\cite{Tarjan83} on red-black trees. In 1992, 
Adams used \funcjoin{} as a primitive for weight-balanced trees to implement set functions for union, intersection,
and difference~\cite{adams1992implementing,adams1993functional}. 
Adams' algorithms were proposed in an international competition for the Standard ML community, which is about
implementations on ``set of integers''. 
Adams won the ``elegance'' award in the competition, meaning that the algorithm is elegant yet reasonably efficient among all participants in the competition. 
However, Adams' original paper only informally argued a loose work bound, which is $O(m+n)$ for operating on two sets with sizes $n$ and $m\le n$. 

Blelloch and Reid-Miller~\cite{Blelloch1998} first considered using join-based algorithms on parallel data structures.
Their paper proposed parallel algorithms on treaps based on \funcjoin{}, and first show that these algorithms have $O(m\log(n/m))$ work and polylogarithmic span (with high probability).
The algorithms are work-efficient because the work matches the lower bound of comparisons needed to combine two ordered sets. 
Later Blelloch et al.~\cite{blelloch2016just} studied the join-based algorithms, and showed work-efficiency and $O(\log n\log m)$ span for other balancing schemes, including AVL trees, read-black trees, and weight-balanced trees. 
In 2020~\cite{blelloch2020optimal}, the span bound for weight-balanced trees was improved to $O(\log n)$. 

In 2015, Akhremtsev and Sanders~\cite{akhremtsev16} proposed join-based algorithms on $(a,b)$-trees. 
To make the algorithm more space-efficient and I/O-friendly in practice, Dhulipala et al. proposed C-trees~\cite{dhulipala2019low} and PaC-trees~\cite{dhulipala2022pac}. The idea is to allow each tree node (or just the leaves) to store a block of entries instead of one. Both C-trees and PaC-trees are binary. 
While their goal is to make the algorithms more cache-friendly in practice, no I/O bounds were shown. 
To the best of our knowledge, our work is the first to study the I/O cost of the join-based algorithms. 

To achieve a better span bound, some techniques of our algorithms are inspired by existing work, such as using a list of pointers based on the height of subtrees~\cite{akhremtsev16}, and using $\sqrt{m+n}$ way divide-and-conquer with a work-inefficient base case~\cite{blelloch2020optimal}. 
Our solution provides non-trivial and effective adaptions to use these ideas to achieve improved I/O bounds on B-trees. 

Due to the elegance of the algorithm, many libraries later 
used the join-based algorithms for their implementation for sequences, sets, or maps 
(e.g., Haskell~\cite{marlow2010haskell}, the LEDA library~\cite{Leda99}, PAM~\cite{sun2018pam},
Aspen~\cite{dhulipala2019low}, CPAM~\cite{dhulipala2022pac}). 

\section{B-Way Join Analysis}
\label{sec:b-way-join-analysis}
We will eventually show the I/O span of the \bjoin{} operation is $O(h_{\max} - h_{\min}),$
and the I/O work is $O(B \cdot (h_{\max} - h_{\min}))$, where $h_{\max}$ and $h_{\min}$ are the maximum and minimal tree heights among the input trees. 
First, we must show that the height of the output tree of each recursive call
does not increase the height of the input trees by more than a constant,
as this could lead to an increase in the I/O work and span.

\begin{lemma}
  \label{lemma:multiway-join-rebalance-bounds}
  Suppose we have an unbalanced B-tree $T$ with $b$ children and $b - 1$ keys.
  Let $C$ be the subset of children of $T$ with height $h$,
  and the rest of the children have height $h - 1$.
  Then we can rebalance the tree to have a valid B-tree of height $h$ or
  $h + 1.$
  If the tree is of height $h + 1$,
  then the root will have at most $|C|$ keys.
\end{lemma}

\begin{proof}
  For each child $c \in C$,
  we can bring the keys up to the root node and reattach the children of $c$
  as children of the root node.
  Suppose $c$ contains $b^\prime$ keys,
  then the number of keys at the root and the number of children both increase
  by $b^\prime.$

  The increase in the number of keys is straightforward.
  The number of children increases by the same amount because bringing $c$ up
  adds $b^\prime + 1$ children,
  and we lost a child by bringing up $c.$

  Performing this operation across all nodes in $C$ will result in an
  at most $|C|B$ keys at the root.
  If the root node also contained $B$ keys initially,
  then there would be $(|C| + 1)B$ keys at the root,
  which would then be split into a tree of height $h + 1,$
  and with at most $|C|$ keys at the root.

  We know that the split will only ever occur once since $|C|$ is upperbounded
  to be at most $B,$
  thus a second split would not occur.
\end{proof}

\begin{lemma}
  \label{lemma:multiway-join-merge-heights}
  Suppose we have a list of trees $T_1, \dots, T_{b + 1}$
  and keys $k_1, \dots, k_{b},$
  such that the height of each tree is either
  $h$ and the root contains at most $B$ keys,
  or $h + 1$ and the root contains at most $1$ key,
  except for $T_1$ which may contain $2$ keys.
  Then we can join the trees such that the node at level $h + 1$ contains at
  most $b + |H_i|$ keys,
  where $H_i$ is a subset of trees $T_1, \dots, T_i$ with height $h + 1.$
  Note that if a tree is of height $h,$ then the number of keys at the node
  of level $h + 1$ is defined to be $0.$
\end{lemma}

\begin{proof}
  We prove the lemma by induction for $i = 1, \dots, b + 1.$
  The base case trivially holds for $i = 1,$
  where if the height of the first tree is $h,$
  then the number of keys at the node of level $h + 1$ is $0.$
  If the height of the first tree is $h + 1,$
  then the number of keys at the node of level $h + 1$ is at most $2,$
  maintaining the stated bound.
  We assume the claim is true for $i,$
  which means the first $i$ trees produces B-tree $T^*$ with at most
  $i - 1 + |H_i|$ keys at the node of level $h + 1.$

  If $T^*$ is of height $h,$
  and the height of $T_{i + 1}$ is $h,$
  then the join will result in at most $1$ key at level $h + 1,$
  if $T^*$ and $T_{i + 1}$ contain $B$ keys each.
  In this case,
  $T^*$ will be made the left child,
  $T_{i + 1}$ will be made the right child,
  and $k_i$ will be the key at the root.

  If $T^*$ is of height $h + 1$ and $T_{i + 1}$ is of height $h,$
  then the join will cause $T_{i + 1}$ to be the right child of $T^*,$
  and $k_i$ will be the right most key of $T^*.$
  This increases the number of keys at the node of level $h + 1$ by $1,$
  and thus the bound is maintained.
  The same logic applies if $T^*$ is of height $h$ and $T_{i + 1}$ is of height
  $h + 1.$

  If instead both trees are of height $h + 1,$
  then the join will cause the trees to be merged at the same height,
  and the number of keys increases by $1$ plus the number of keys in $T_{i + 1}.$
  Recall in the statement of the lemma that every tree except for $T_1$ contains
  at most $1$ key at the root.
  Therefore,
  we can bound the increase by at most $2$ keys,
  and we can see the bound is maintained since $i$ increase by $1$
  and $|H_i|$ also increases by $1.$
\end{proof}

\begin{theorem} \label{theorem:multiway-join-height}
  Let $T_1, \dots, T_{b + 1}$ and $k_1, \dots, k_{b}$
  be the set of input trees and keys respectively,
  where $b \leq B.$
  If $h$ is the height of the tallest tree in $T,$
  then after performing \bjoin$(T, k),$
  the height of the output tree will be either $h$ or $h + 1.$
\end{theorem}

\begin{proof}
  We prove the theorem by induction. In the base case,
  if all trees in $T_{1, \dots, b + 1}$ are of height $h,$
  joining them maintains the theorem trivially.
  We create a new root node $v$ containing all keys of $k,$
  and each tree in $T_{1, \dots, b + 1}$ is made a child of $v.$ 
  If any child node contains less than $B / 2$ keys,
  then we can use a standard rebalancing algorithm to fix it. 
  This may reduce the number of keys in $v$ to less than $b$,
  but it will never increase. 
  If all the keys in $v$ are removed,
  then the output tree will also be of height $h.$
  Otherwise, it will be $h + 1.$

  Then we prove the inductive step.
  Let $T^\prime_{0,1,\dots}$ be the set of trees which contain children that are the result
  of a recursive call, i.e. the tallest trees in $T_{1, \dots, b + 1}.$
  Due to the inductive hypothesis,
  all the trees will have a right most child of height either $h - 1$ or $h.$
  The first tree may also have a left most child of height $h.$
  Due to the possible discrepancy in heights,
  we must rebalance these trees.
  We know from Lemma \ref{lemma:multiway-join-rebalance-bounds},
  that we can rebalance these trees to have heights of $h$ or $h + 1.$
  If the height of a tree in $T^\prime$ is $h + 1,$
  then we know there is at most $1$ key at the root,
  except for $T^\prime_0$ which may have $2$ keys at the root.

  According to Lemma \ref{lemma:multiway-join-merge-heights},
  we can see that merging the trees in $T^\prime$ will result in a tree of
  height $h$ or $h + 1.$
  If the tree is of height $h,$
  then the proof is complete.
  Otherwise,
  we must ensure the root node has at most $B$ keys.
  We know the root node has at most
  $b^\prime + |H_{b^\prime}|$ keys,
  where $b^\prime = |T^\prime|.$
  We know that $b^\prime + |H_{b^\prime}| \leq b,$
  since in order for a tree to have height $h + 1,$
  it must have been passed into the recursive call with at least one other tree
  as part of the input set.
  Therefore,
  in order for $|H_{b^\prime}|$ to increase by $1,$
  the value of $b^\prime$ must have decreased by at least $1.$
  Therefore,
  the number of keys at the root is upper bounded by $b,$
  and thus cannot overflow again,
  maintaining the height of the tree $h + 1.$
\end{proof}

\multijoincost*

\begin{proof}
  We can see each separator key is assigned a thread, which is then responsible
  for traversing down a pair of input trees until it reaches the point where all
  trees are of the same height,
  which is a traversal of length of at most $O(h_{\max} - h_{\min}).$

  At each level of the tree,
  each thread will perform at most $O(1)$ I/Os.
  We know from \cref{theorem:multiway-join-height} at each level of the tree,
  the resulting subtree will always stay either the same height or increase by
  1.
  Which means no additional I/Os must be performed in order to compensate for
  the increased height of the tree,
  resulting in an I/O span of $O(h_{\max} - h_{\min}).$

  The I/O work comes readily from the I/O span.
\end{proof}

\section{Multiway Split: Cost Analysis}
\label{sec:multiway-split-cost}

\multiSplitCost*

\begin{proof}
  The algorithm spawns $d + 1$ threads,
  where each thread performs a search down the tree.
  Each thread will perform $O(\log_B n)$ I/Os,
  in order to perform the search and copying the keys and
  pointers to the children takes $O(1)$ I/Os per node.

  When performing the $\bjoin{}$ operation,
  recall that the I/O span of the $\bjoin{}$ operation is $O(h_{\max} - h_{\min}),$
  which is the difference in height between the tallest and shortest trees
  being joined together,
  and that the output tree must be of height at least $h_{\max}.$

  We will be performing \bjoin{} along the path from the first divergence to
  each key.
  This means if at least one key is within the range of the split,
  then every \bjoin{} occurs with trees with height difference at most $O(1),$
  and satisfying the I/O span bound of $O(\log_B n).$

  For the cases where a node contains keys within the range,
  but some of their ancestor's does not,
  we may use an amortized argument to charge a constant I/O for each node which
  does not contain keys within the range,
  which would otherwise not incur any cost.
  Then when joining with trees of height difference larger than a constant,
  we use the charge to pay for the extra I/O.
  The height difference can only be at most the number of consecutive skipped
  \bjoin{} operations.

  For the I/O work bounds,
  observe that for each \bjoin{} operation,
  at most two trees being joined are of different heights,
  which are the trees along the search path to the split key.
  Therefore,
  the I/O work bounds is $O(1)$ instead of $O(B)$ for each \bjoin{} operation,
  which gives us the I/O work bound of $O(d \log_B n)$ I/Os.
\end{proof}

\section{Multiway Join Cost Analysis}
\label{sec:multi-join-cost-analysis}
\begin{lemma}
  \label{lem:preprocess-cost}
  The cost of preprocessing $d$ B-trees such that each tree $T$ contains a list
  of pointers to all nodes along their left and right spines is
  $O(\log_B n + \log_B d)$ I/O span and $O(d \log_B n)$ I/O work,
  where $n$ is the total number of keys in all trees.
\end{lemma}

We can see that the fuse operations are completely independent of each other,
and since each thread is able to access the number of keys that will be copied
to the same tree ahead of them,
then this process is able to occur in $O(1)$ I/O span and $O(B)$ I/O work,
which gives us the following lemma.

\begin{lemma}
  \label{lem:multiway-join-fuse-cost}
  The cost of fusing $B$ trees in the \bjoinfast{} operation is
  $O(1)$ I/O span and $O(B)$ I/O work.
\end{lemma}

Next we show the following lemma,
so that we can guarantee that after the first round of divisions,
no overloaded node will push more than a single key up to their parent node.

\begin{lemma}
  \label{lem:multiway-join-divide-up-bound}
  After fusing $B$ trees together in the \bjoinfast{} operation,
  we can guarantee after one round of dividing the overloaded nodes,
  that each subsequent overloaded node may push up at most one key to their
  parent node.
\end{lemma}

\begin{proof}
  Recall that we are joining at most $B$ trees together.
  Since the root and all nodes it may fuse with contains at most $B$ keys,
  then including the paired key which the tree is associated with,
  each tree contributes at most $1$ key to be pushed up to their parent node.
  In other words,
  for every tree that a node fuses with,
  they contribute at most $1$ key to push to the parent per node.
  
  Therefore,
  after a single round of divisions,
  even in the worst case where all keys were pushed up to a single node,
  this node would contain at most $2B$ keys,
  which would only push up a single key to the parent.
\end{proof}

\begin{lemma}
  \label{lem:multiway-join-rebalance-bounds}
  The cost of dividing the nodes up the tree in the \bjoinfast{} operation is
  $O(\log_2 \log_B n)$ I/O span and
  $O(B\log_B n)$ I/O work.
\end{lemma}

\begin{proof}
  The first set of divide operations cost $O(1)$ I/O span and $O(B)$ I/O work,
  since at most $B$ new nodes are created,
  which can all be done in parallel.

  Then we gather the number of keys in each node along the path from the root
  to every overloaded node.
  This takes $O(\log_B \log_B n)$ I/O span and $O(B \log_B n)$ I/O work,
  as we can use the list of pointers to jump directly to their assigned node
  along the path,
  since the path must consist of only nodes which were along the spines of a
  fused tree.
  There are at most $B$ overloaded nodes,
  which means we create $B$ arrays for the paths to the root.

  We then perform the prefix sum operation along each of these arrays,
  which takes $O(\log_2 \log_B n)$ I/O span and $O(B \log_B n)$ I/O work.
  Finally,
  using the information provided by the prefix sum,
  each node can be divided in a $O(1)$ I/O for each thread,
  which gives us our desired results.
\end{proof}

\begin{lemma}
  \label{lem:multiway-join-update-pointers}
  The cost of updating the list of pointers in the \bjoinfast{} operation is
  $O(\log_B \log_B n)$ I/O span and $O(B \log_B \log_B n)$ I/O work.
\end{lemma}

\begin{proof}
  Each list of pointers is a standard sequential B-tree.
  Therefore,
  the split operation takes $O(\log_B \log_B n)$ I/O span and work for each
  list.
  There are $B$ lists of pointers,
  which gives us $O(B \log_B \log_B n)$ I/O work.

  Then by using the \bjoin{} operation from Section \ref{sec:b-way-join},
  we can join these trees in $O(\log_B \log_B n)$ I/O span and
  $O(B \log_B \log_B n)$ I/O work.
\end{proof}

Now we can prove the main theorem of this section.

\begin{theorem}
  \label{thm:multiway-join-cost}
  We can join $T_1, \dots, T_d$ B-trees and $k_1, \dots, k_d$ keys together in
  parallel with $O(\log_B d \cdot \log_2 \log_B n + \log_B n)$ I/O span and
  $O(d \log_B n)$ I/O work, where $n$ is $\sum_{i = 1}^d |T_i|,$
  in the \ourmodel{}.
\end{theorem}

\begin{proof}
  We get from Lemma \ref{lem:preprocess-cost} that the cost of preprocessing
  is $O(\log_B n + \log_B d)$ I/O span and $O(d \log_B n)$ I/O work.

  The \mjoin{} operation performs $O(\log_B d)$ rounds of \bjoinfast{}
  operations.
  This and preprocessing gives us the I/O span of
  $O(\log_B d \cdot \log_2 \log_B n + \log_B n).$

  To analyze the I/O work cost,
  we can see that the total number of \bjoinfast{} operations reduces by a
  factor of $B$ each round.
  Therefore,
  we calculate the total I/O work cost from the following equation
  $$
    \sum_{i = 1}^{\log_B d} B^i \log_B n = O(d \log_B n).\\
  $$
\end{proof}

\section{Prefix Sum}
\label{sec:prefix-sum}
The prefix sums problem is a fundamental aspect of parallel computing.
The applications for this problem extends to a wide variety of domains,
including for the Parallel Joinable B-Tree.

\begin{definition}
  Given an ordered set $A$ of $n$ elements,
  the prefix sums operation returns an ordered set $A^\prime$ of $n$ elements,
  such that $B[i] = \sum_{j = 0}^i A[j]$ for all $0 \leq i < n.$
\end{definition}

We show that if the input set $A$ is located in contiguous main memory,
then the prefix sums problem can be solved in the \ourmodel{} in $O(\log_2 n)$
I/O span and $O(n)$ I/O work.
By using the \threadgather{} operation,
$n$ non-contiguous elements can be gathered into a single contiguous array in
$O(\log_B n)$ I/O span and $O(n)$ I/O work,
thus the above result also holds for non-contiguous input sets.

\begin{lemma}
  The prefix sum problem can be solved in the \ourmodel{} in $O(\log_2 n)$ I/O
  I/O span and $O(n)$ I/O work.
\end{lemma}

\begin{proof}
The \ourmodel{} solution to the prefix sums problem is performed by simulating
the optimal I/O work efficient PRAM algorithm.
The algorithm performs two passes: building a tree bottom-up and then
traversing the tree top-down.
Each element in $A$ is a leaf node,
ordered from left to right,
and we combine the elements in pairs to form a binary tree,
where each node is the sum of its two children.
Naturally this costs $O(\log_2 n)$ I/O span and $O(n)$ I/O work.
Then from top down the tree,
we pass the sum of the left child to the right child,
summing up the left child's values until we reach the leaf,
which then sums up the values of the left children passed.
This top-down traversal also costs $O(\log_2 n)$ I/O span and $O(n)$ I/O work.

The tree can be implemented as an array,
which would take advantage of the locality of the data.
However,
for the sake of simplicity,
we will not use this optimization in our implementation as it does not affect
the bounds of our algorithm.
\end{proof}

\section{Union Algorithm: Base Case Cost}
\label{sec:union-base-case-cost}

\begin{lemma}
  The \union{} algorithm takes $O(\log_B m \cdot \log_2 \log_B n + \log_B n)$
  I/O span and $O(m \log_B n)$ I/O work.
\end{lemma}

\begin{proof}
  The cost of the \union{} algorithm is dominated by the \mjoin{} operation,
  which comes readily from \cref{thm:multiway-join-cost},
  where $d = m.$
  This gives us the I/O span $O(\log_B m \cdot \log_2 \log_B n + \log_B n)$
  and I/O work $O(m \log_B n).$
\end{proof}

\unionBaseCaseCost*

\begin{proof}
  Suppose we have $k$ base cases,
  then the total cost of all base cases in the union algorithm is
  $\sum_{i = 1}^k O(m_i (\log_B n_i + B)),$
  where $m_i$ and $n_i$ are the number of keys in the $i$-th base case.
  Recall that for each base case,
  $m_i < \sqrt{n_i + m_i}.$

  From Lemma B.2 from Blelloch et al. \cite{blelloch2019optimal},
  we know that $\sum_{i = 1}^k m_i \log n_i = O(m \log n),$
  where $\sum_{i = 1}^k m_i = m, \sum_{i = 1}^k n_i = n,$ and for all $i,$
  $m_i \leq \sqrt{n_i + m_i}.$
  In the proof,
  the base of the logarithm is independent of the result,
  thus we can use the same proof to show that the I/O work is
  $O(m \log_B (\frac{n}{m})).$
\end{proof}

\section{Additional Pseudocode for Multiway Split and Join}
We present the pseudocodes in \cref{algo:msplit,algo:mjoin,algo:dividenode,algo:bjoinfast}.

\begin{algorithm}[H]
\small
\caption{$\msplit{}(T, k_{1..d})$\label{algo:msplit}}
\KwIn{A B-trees $T$ and a sorted list of $d$ keys $k_{1..d}$}
\KwOut{A B-tree $T$ with all keys in $\{k_1, k_2, \dots, k_d\} \cup \bigcup T_i$}
\parForEach{$1 \leq i \leq d + 1$}{
  $T_i \gets \tsplit{}(T, k_{i - 1}, k_i)$
}
\Return{$T_{1..d + 1}$}
\end{algorithm}

\begin{algorithm}[H]
\small
\caption{$\mjoin{}(T_{1..d+1},k_{1..d})$\label{algo:mjoin}}
\KwIn{A sequence of $d+1$ B-trees $T_{1..d+1}$ and $d$ keys $k_{1..d}$.
For any $x\in T_{i}$ and $y\in T_{i+1}$, we have $x<k_i<y$.}
\KwOut{A B-tree $T$ with all keys in $\{k_1, k_2, \dots, k_d\} \cup \bigcup T_i$}
\parForEach{$1 \leq i \leq d + 1$}{
  $\ell, r \gets$ the left most and right most leaf node of $T_i$\\
  $L_i, R_i \gets$ empty list\\
  \While{$\ell$ and $r$ are not null}{
    Add $\ell$ to $L_i$ and $r$ to $R_i$\\
    $\ell, r \gets$ parent of $\ell$ and $r$\\
  }
}

\ForEach{$1 \leq \ell \leq \lceil \log_B d \rceil$}{
  \parForEach{$0 \leq i < \frac{d}{B^\ell}$}{
    $T_i \gets \bjoinfast{}(T_{Bi + 1, \dots, Bi + B}, k_{Bi + 1, \dots, Bi + B - 1})$\\
  }
}
\Return{$T_1$}
\end{algorithm}

\begin{algorithm}[H]
\small
\caption{$\joinwithtall{}(T_{1..b+1}, k_{1..b})$\label{algo:joinwithtall}}
\DontPrintSemicolon
\KwIn{
  A sequence of $b+1$ B-trees $T_{1..b+1}$ and $b$ keys $k_{1..b}$.\\
  For any $x\in T_{i}$ and $y\in T_{i+1}$, we have $x<k_i<y$.\\
  For all $T_i,$ there is a list $L_i$ and $R_i$ of pointers to all nodes along $T_i$'s left and right spines.\\
}
\KwOut{A B-tree $T$ with all keys in $\{k_1, k_2, \dots, k_b\} \cup \bigcup T_i$}
\tcp{Step 2: Fusing Within Groups}
\parForEach{$2 \leq i \leq b + 1$}{
  $j \gets$ the largest index such that $h(T_j) \geq h(T_i)$ and $j < i$\tcp*[f]{$T_i$ will join with $T_j$}\\ 
  $\treejoinset{} \gets $ all trees $T_x$ which join with $T_j$ and $x < i$\\
  $n \gets \sum_{T \in X}$ number of keys at the root of $T$\\
  $v_i \gets$ the root of $R_{j}[h(T_i)]$\\
  Copy values at root of $T_i$ and $k_{i - 1}$ to $v$ offset by $n$\\
}

\tcp{Step 3: Dividing Large Nodes}
\parForEach{$1 \leq i \leq b + 1$}{
  \If{$v_i$ is overloaded}{
    $k^\prime_{i, [1, \dots, b]} \gets \dividenode{}(v_i)$\\
  }
}

$X \gets \{\}$\tcp*[f]{Set of nodes which must still be divided}\\
\parForEach{$1 \leq i \leq b + 1$}{
  \If{$k^\prime_{i, [1, \dots, b]}$ exists}{
    $p_i \gets$ the parent of $v_i$\\
    Add $k^\prime_{i, [1, \dots, b]}$ to $p_i$\\
    $X \cup p_i$ if $p_i$ is overloaded\\
  }
}

\tcp{Step 4: Dividing Smaller Nodes in Parallel}
$C_{1, \dots, |X|} \gets$ empty arrays of length $\log_B n$\\
\parForEach{$v \in X$}{
  \parForEach{$1 \leq i \leq \log_B n$}{
    $u \gets$ the node $i$ levels above $v$\\
    $C_v[i] \gets$ 1 if $u$ is full, 0 otherwise\\
  }
}

\parForEach{$v \in X$}{
  \prefixsum{}$(C_v)$\\
}

\parForEach{$v \in X$}{
  \parForEach{$1 \leq i \leq \log_B n$}{
    $u \gets$ the node $i$ levels above $v$\\
    \lIf{$C_v[i] = i$}{
      Mark $u$ to be divided\tcp*[f]{Multiple threads may mark $u$ to be divided}\\
    }
  }
}

\parForEach{$v \in X$}{
  \parForEach{$1 \leq i \leq \log_B n$}{
    $u \gets$ the node $i$ levels above $v$\\
    \If{$u$ marked to be divided}{
      $y \gets$ number of threads which marked $u$\\
      Get one key from $y$ of its children\\
      \dividenode{}(u)\\
    }
  }
}

\Return{$T_1$}
\end{algorithm}

\begin{algorithm}[H]
\small
\caption{$\dividenode{}(v)$\label{algo:dividenode}}
\DontPrintSemicolon
\KwIn{A node $v$ containing more than $B$ key}
\KwOut{Node $v$ is divided into $b$ valid B-tree nodes, and outputs $b - 1$ keys}
\lIf{$v$ is root}{
  $p \gets$ an empty parent node
}
\lElse{
  $p \gets$ the parent of $v$
}
$b \gets \lceil \frac{|v|}{B} \rceil$\tcp*[f]{$|v|$ is the number of keys in $v$}\\
$b^\prime \gets \lceil \frac{|v|}{b} \rceil$\\
\parForEach{$i \in [1, b]$}{
  $v_i \gets$ a new node containing keys from $v$ in the range $[(i - 1) \cdot b^\prime + 1, i \cdot b^\prime)$\\
}
\Return{Every $i \cdot b^\prime$-th key in $v$ for all $i \in [1, b)$}
\end{algorithm}

\section{Additional Set Operations}
\label{sec:additional-set-operations}

The \intersection{} and \difference{} operations are similar to the
  \fastunion{} algorithm under the join-based framework.
Hence, due to the space limit, we postpone the details of these two operations here in the appendix.
The main difference lies in how subtrees are joined back.
  
\begin{theorem}
    Given two B-trees with sizes $m$ and $n\ge m$,
    there exists a parallel algorithm that returns a new B-tree containing the
    intersection and difference of the two input trees in
    $O\left(m\log_B\left(\frac{n}{m}\right)\right)$ I/O work,
    $O(\log_B m \cdot \log_2 \log_B n +\log_B n)$ and I/O span,
    where $B$ is the block size. 
\end{theorem}

\begin{proof}
  We first modify the \mjoin{} operation slightly to take in an array $d$ boolean
  values,
  where each boolean value indicates whether the corresponding key should be
  included in the final result.
  Suppose $k_i$ is not to be part of the final result.
  Then we instead find the largest key in $T_i,$
  denoted as $k^\prime_i,$
  and split on $T_i$ using $k^\prime_i.$
  Then we replace $k_i$ with $k^\prime_i,$
  and perform the join operation as normal.

  In the \intersection{} operation,
  after searching for the $d - 1$ split keys,
  we check if all $d - 1$ keys are in both trees.                                                                   
  If a key is in both trees,
  then we mark the key to be included in the final result.
  Then we perform the \msplit{} operation as normal,
  make our recursive calls,
  then use the modified \mjoin{} operation to combine the results,
  excluding any keys which did not appear in both trees.

  The base case of the \intersection{} operation is also very similar to the
  \union{} algorithm.
  Let $|T_1| \geq |T_2|,$
  as we can always swap the two trees if this is not the case.
  Then we use the keys from $T_2$ to search for the keys in $T_1.$
  Any keys which are not found in $T_1$ are discarded.
  All remaining keys are then joined together using the traditional \mjoin{}
  operation.

  For the \difference{} operation,
  if a split key is found to be in $T_1$ and $T_2$ during the main recursive
  algorithm,
  then it is marked as excluded,
  and all the same steps as the \intersection{} operation are performed.
  For the base case,
  we use a slightly different approach depending on which tree is larger.
  If $|T_1| \geq |T_2|,$
  then we split $T_1$ using the keys in $T_2,$
  and join with none of the keys from $T_2.$
  Otherwise,
  we use all keys in $T_1$ to search for the keys in $T_2,$
  and join all keys from $T_1$ which were not found in $T_2.$

  As such, the same analysis used to get the bounds from the \fastunion{} algorithm
  can also be applied to the \intersection{} and \difference{} operations and achieve the same bounds 
 .
\end{proof}

\end{document}